%% file: adapense.tex
\documentclass[12pt]{article}
\usepackage[inner=.75in,outer=.75in,top=1.33in,bottom=1.5in]{geometry}
\usepackage{graphicx}
\usepackage[font=footnotesize,labelfont=bf,indent]{caption} 
\usepackage{authblk}
\usepackage{xcolor}
\usepackage{csquotes}
\usepackage[english]{babel}
\usepackage{hyperref} 
\usepackage{booktabs}
\usepackage{enumitem} 
\usepackage[style=authoryear,useprefix=true,url=false,sorting=nyt,mincitenames=1,maxcitenames=2,maxbibnames=9,citetracker=true,backend=biber,uniquelist=false,giveninits=true,uniquename=false]{biblatex}
\usepackage{setspace}
\usepackage{subfigure}
\usepackage{appendix}
\usepackage{xr} 
\usepackage{refcount} 
\newbibmacro*{journal}{%
  \iffieldundef{journaltitle}
    {}
    {\printtext[journaltitle]{%
       \printfield[noformat]{journaltitle}%
       \setunit{\subtitlepunct}%
       \printfield[noformat]{journalsubtitle}}}}

\DeclareFieldFormat{titlecase}{\MakeSentenceCase{#1}} 
\input{common-cmds.tex}

\graphicspath{{./figures/}{./}}

\makeatletter
\newcommand*{\storecounter}[2]{%
  \edef\@currentlabel{\the\value{#1}}
  \label{#2}
}
\makeatother

\begin{document}

\title{Robust Variable Selection and Estimation Via Adaptive Elastic Net S-Estimators for Linear Regression}

\author{David Kepplinger}
\affil{Department of Statistics, School of Computing, George Mason University}
\date{July 2021}

\renewcommand\Affilfont{\itshape\small}

\maketitle

\abstract{
Heavy-tailed error distributions and predictors with anomalous values are ubiquitous in high-dimensional regression problems and can seriously jeopardize the validity of statistical analyses if not properly addressed.
For more reliable estimation under these adverse conditions, we propose a new robust regularized estimator for simultaneous variable selection and coefficient estimation.
This estimator, called adaptive PENSE, possesses the oracle property without prior knowledge of the scale of the residuals and without any moment conditions on the error distribution.
The proposed estimator gives reliable results even under very heavy-tailed error distributions and aberrant contamination in the predictors or residuals.
Importantly, even in these challenging settings variable selection by adaptive PENSE remains stable.
Numerical studies on simulated and real data sets highlight superior finite-sample performance in a vast range of settings compared to other robust regularized estimators in the case of contaminated samples and competitiveness compared to classical regularized estimators in clean samples.
}

\section{Introduction}
Simplicity of the linear regression model ensures its continued importance in many scientific and industrial applications, especially if a small sample size prohibits use of more complex models.
This paper considers prediction and variable selection in the linear regression model
\begin{equation}\label{def:linear-regression-model}
\randY = \mutrue + \randX \tr \btrue + \randE
\end{equation}
with $p$-dimensional random predictors $\randX$ independent of the random error term $\randE$ and fixed parameters $\mutrue \in \mathbb{R}$, $\btrue \in \Rm{p}$.
Based on a sample of $n$ independent realizations of $\randY$ and $\randX$, collected as pairs $(y_i, \mat x_i)$, $i = 1, \dotsc, n$, the statistical goal is to estimate the intercept $\mutrue \in \mathbb{R}$ and slope coefficients $\btrue \in \Rm{p}$.
Emphasis is on high prediction accuracy and identification of relevant predictors, i.e., those which have non-zero entries in $\btrue$.

With the ever growing abundance of data, often combined from various sources, it is increasingly challenging to make assumptions on the distribution of the error term $\randE$ or assume that the data, particularly the predictors, are free of anomalous values or gross outliers.
In proteomics or genomics studies, for instance, undetected equipment failure, problems with sample preparation, or patients with rare phenotypic profiles, are just a few sources of contamination that can severely damage the effectiveness of most statistical methods commonly used to estimate the parameters in the linear regression model.
If potential contamination or heavy tailed errors are not properly addressed, they can jeopardize the validity of statistical analyses and render the results unreliable.
The main goal of this work is therefore to develop a method for reliable identification of the relevant predictors and estimation of the corresponding non-zero regression coefficients in the linear model~\eqref{def:linear-regression-model} for heavy-tailed errors and under the potential presence of contaminated observations in the sample.

Many methods for identifying the set of relevant predictors in the linear regression model have been proposed under the assumption of a (sub-)Gaussian error distribution and well-behaved predictors.
Most of these methods simultaneously identify relevant predictors and estimate their coefficients.
The dual formulation leads to the common regularized regression objective
\begin{equation}\label{eqn:regularized-loss}
\argmin_{\mu \in \mathbb R,\mat\beta \in \Rm{p}} \loss(\mat y, \mu + \mat X \mat\beta) + \lambda \Phi(\mat\beta),
\end{equation}
where $\loss(\mat y, \mu + \mat X \mat\beta)$ is a regression loss (e.g., the sum of squared residuals) and $\Phi(\mat\beta)$ is a penalty function (e.g., the $L_1$ norm).
The LASSO \parencite{Tibshirani1996} and the elastic net estimator \parencite{Zou2005} are prominent examples of regularized regression estimators based on the least-squares (LS) loss.
While the regularized LS-loss is extensively studied and well understood under numerous settings and penalty functions, less has been done to enable variable selection and efficient coefficient estimation for heavy-tailed error distributions and under the potential presence of contamination in the predictors and the response.

Several regularized estimators which promise more resilience towards adverse contamination have been proposed over the years.
Most proposals replace the convex LS-loss function with a robust alternative, i.e., a loss function which is less affected by contamination and outliers.
A recently very active stream of research \parencite{Wang2007b,Fan2014b,LambertLacroix2011,LambertLacroix2016,Zheng2017,Sun2019,Fan2018,Loh2018f,Pan2020f} promotes convex loss functions which increase slower than the LS-loss for larger residuals.
These so-called ``unbounded M-loss'' functions (e.g., the sum of absolute deviations or the Huber loss), are designed to shield against heavy-tailed error distributions in high-dimensional settings.
While the convexity of unbounded M-loss functions enables derivations of strong theoretical guarantees, these estimators are still exposed to the potentially devastating effects of contamination in the numerous predictors.
Commonly suggested remedies, e.g., down-weighting observations with ``unusual'' predictor values or univariate winsorizing \parencite{Loh2017,Sun2019}, are ill-suited for high dimensional problems.
In the sparse estimation regime, for example, down-weighting observations due to outlying values in irrelevant predictors (which are unknown in advance) may sacrifice precious information.

Regularized M-estimators using a bounded, and hence non-convex, M-loss function yield the desired protection against contaminated predictors.
Asymptotic properties of these regularized bounded M-estimators have been recently studied.
\textcite{Smucler2017}, for instance, propose the MM-LASSO estimator; a regularized M-estimator relying on an auxiliary M-estimate of the residual scale (hence ``MM'').
The authors derive the oracle property for their MM-LASSO estimator under fixed dimensionality but otherwise very general conditions.
\textcite{Loh2017}, on the other hand, proves oracle bounds for the estimation error for M-estimators regularized by folded-concave penalties (e.g., SCAD), by restricting the problem~\eqref{eqn:regularized-loss} to a neighborhood around the origin which must contain the true parameters. 
For these oracle bounds to hold, however, the loss function, given the sample, must satisfy strict conditions, including restricted strong convexity in a neighborhood of the true parameters.
The main challenge for using regularized M-estimators in practice, however, is the requirement for a robust estimate of the scale of the residuals.
Obtaining such a robust estimate with sufficiently small finite-sample bias is a difficult task in its own right; almost insurmountable in high dimensional problems with contamination and heavy-tailed errors.

Recently renewed attention has been given to the mean-shift outlier model (MSOM), introducing an additive nuisance parameter for each observation to quantify its outlyingness.
To identify these outliers, \textcite{She2021} add a constraint on the number of non-zero nuisance parameters, i.e., on the number of outlying observations.
With this formulation, they develop a theoretical foundation for a broad class of robust estimators defined algorithmically, e.g., the popular SparseLTS estimator \parencite{Alfons2013}.
Under a very general framework allowing for different loss and penalty functions, \textcite{She2021} establish minimax bounds for the estimation error under the MSOM and propose an efficient algorithm for estimation.
In a similar sprit, \textcite{Insolia2020} propose a mixed integer program (MIP) to constrain both the number of outliers and the number of relevant predictors using the $L_0$-pseudo-norm.
The authors develop guarantees for the algorithmic complexity and the statistical estimation error for this MIP under the MSOM, even for ultra-high dimensional problems.
For example, with normally distributed errors and when the true number of relevant predictors as well as the true number of outliers in the response are known, the MIP possess the ``robustly strong oracle property'', meaning their method possess the oracle property under the MSOM.
The MSOM framework and the proposed methods building upon the $L_0$-pseudo-norm provide a promising avenue, particularly in regimes with high signal strength.
For lower signal strengths, however, the $L_0$-pseudo-norm for either variable selection or outlier detection tends to suffer from high variability.
For variable selection and prediction, for example, other penalties may deliver better performance \parencite{Hastie2017,Insolia2020}.
In addition, for computational tractability and stability, it is necessary to restrict the optimization to a tight neighborhood around the true regression parameters, as well as to have a good understanding of the number of relevant predictors and the number of outliers.
\textcite{Insolia2020}, for example, suggest to use an ensemble of other robust regression estimates to get preliminary estimates of the necessary bounds.
In this paper we are proposing a method which does not require any prior knowledge of the number of truly relevant predictors and only a rough upper bound on the number of outliers.
Therefore, our method is a good candidate to initialize the methods proposed for the MSOM.

This paper introduces a method which achieves good asymptotic properties as well as strong empirical performance without requiring any prior knowledge about the true parameters, the scale of the residuals, or the number of relevant predictors.
We tackle this problem by building upon the ``S-loss'' function, a loss function based on a robust measure of the scale (hence S-loss) of the fitted residuals, circumventing the need for an auxiliary estimate of the residual scale.
In the unpenalized case, the S-estimator is highly robust towards heavy-tailed errors and arbitrary contamination in the predictors \parencite{Rousseeuw1984}.
So far, only a handful of regularized S-estimators have been proposed \parencite{Maronna2011, Gijbels2015, CohenFreue2019} and the theoretical guarantees are not yet well established.
While the S-loss combined with the $L_2$ penalty (S-Ridge) \parencite{Maronna2011} does not lead to a sparse estimator $\hmat\beta$, \textcite{Smucler2017} show that it is root-n consistent for random predictors and only weak conditions on the error term if the dimension is fixed.
The authors leverage this root-n consistency of the S-Ridge and use the robust M-scale of the fitted residuals in the formulation of their (adaptive) MM-LASSO estimator.
These results are encouraging as they address the issue of estimating the residual scale for redescending M-estimators in high dimensional settings, but the finite-sample bias of the M-scale of the residuals often undercuts the good theoretical properties in practice.
Under a fixed design, \textcite{CohenFreue2019} show that the S-loss combined with a sparsity-inducing elastic net penalty (called PENSE) leads to a consistent estimator for the true regression parameters even for heavy-tailed error distributions.
All of these results for S-estimators are obtained without reliance on a residual scale estimate.

The first main contribution of this paper is the introduction of adaptive PENSE, a regularized S-estimator combining the S-loss with an adaptive elastic net penalty.
The results presented here show that the adaptive PENSE estimator possesses the oracle property under similar conditions as in \textcite{Smucler2017}.
For deriving these result we also extend the theory pertaining to PENSE \parencite{CohenFreue2019}.
Leveraging a PENSE estimate, the adaptive penalty used in this work reduces the bias of coefficients of relevant predictors and screens out many irrelevant predictors.
While asymptotic results require fixed dimensionality of the predictor matrix, the only other condition on the predictors are finite second moments.
Importantly, no moment conditions on the error distribution are required and the estimator is completely free from tuning to an unknown error distribution.
Therefore, the results apply equally to light- and heavy-tailed error distributions, including the Cauchy distribution and other symmetric stable distributions.

The second main contribution of this paper is to describe scalable and reliable algorithms to compute adaptive PENSE estimates, even for high dimensional data sets.
Computation of adaptive PENSE estimates is challenging due to the highly non-convex objective function and several hyper-parameters.
Non-convexity necessitates strategies for selecting suitable starting-points for numerical algorithms to locate minima of the objective function.
Building upon the work in \parencite{CohenFreue2019}, novel computational strategies are proposed to increase exploration of the parameter space while retaining computational feasibility.
To ensure a large range of problems are amenable to adaptive PENSE the optimized algorithms are made available in an easy-to-use R package.

The paper is organized as follows.
The adaptive PENSE method is described in detail in Section~\ref{sec:method}.
Relevant for practical applications, Section~\ref{sec:computing} outlines the algorithms for computing adaptive PENSE estimates and provides a resilient strategy for choosing hyper-parameters.
Section~\ref{sec:theory} presents the main theoretical results pertaining to the robustness and oracle properties of the adaptive PENSE estimator, along with a discussion of the imposed assumptions.
Section~\ref{sec:numerical-studies} outlines the strong empirical performance of adaptive PENSE in simulation studies and real-world applications.
Supporting lemmas and proofs of the theorems, as well as additional simulation results, are provided in the Supplementary Materials.

\subsection{Notation}
To simplify the following exposition some notation is fixed throughout.
The concatenated parameter vector of intercept and slope coefficients in the linear regression model~\eqref{def:linear-regression-model} is denoted by $\mat\theta = (\mu, {\mat\beta}\tr)\tr$.
The non-zero elements of a slope parameter $\mat\beta$ are referenced as $\mat\beta_\subRnum{1}$, while the zero elements are written as $\mat\beta_\subRnum{2}$.
Adaptive PENSE estimates are always denoted by a circumflex, $\hmat\theta$, and PENSE estimates are marked by a tilde, $\tmat\theta$.
The subscript $i \in \{1, \dotsc, n \}$ is exclusively used to denote the $i$-th observation from the sample, while $j \in \{1, \dotsc, p \}$ is reserved for indexing predictors.
Without loss of generality, it is assumed that the true slope parameter equals the concatenated vector $\btrue = {({\btrue_\subRnum{1}}\tr, {\btrue_\subRnum{2}}\tr)}\tr$ where the first $s$ elements, $\btrue_\subRnum{1}$, are non-zero and the trailing $p - s$ elements are zero, i.e., $\btrue_\subRnum{2} = \mat 0_\subt{p-s}$.

\section{Adaptive PENSE}\label{sec:method}

We propose estimating the sparse regression parameter $\trueparam$ in the linear regression model~\eqref{def:linear-regression-model} by penalizing the robust S-loss with an adaptive elastic net penalty.
In the presence of gross errors in the response variable (outliers) and unusual values in the predictors (leverage points), the highly robust S-loss is an appropriate surrogate for the LS-loss.
The S-loss is given by
\begin{equation}\label{eqn:s-loss}
\mathcal{O}_\subt{S} (\mat y, \hmat y) = \mhscalesq{\mat y - \hmat y} = \inf \left\{ s^2 \colon \frac{1}{n} \sum_{i = 1}^n \rho \left(\frac{y_i - \hat y_i} {|s|} \right) \leq \delta \right\},
\end{equation}
where $\rho$ is a bounded and hence non-convex function and $\delta \in (0, 0.5]$ is a fixed parameter governing robustness properties as will be shown later.

Instead of the classical variance of the fitted residuals, the S-loss minimizes the square of the robust M-scale of the fitted residuals, $\mhscale{\mat y - \hmat y}$.
If the number of exactly fitted observations $\# \{ i\colon y_i = \hat y_i\} < n (1 - \delta)$, the M-scale estimate is greater than 0 and is given implicitly by the equation
$$
\frac{1}{n} \sum_{i = 1}^n \rho \left(\frac{y_i - \hat y_i} { \mhscale{\mat y - \hmat y} } \right) = \delta.
$$
To ease notation, we define the M-scale of the residuals of an estimate $\hmat\theta$ by $\mhscale{\hmat\theta} = \mhscale{\mat y - \hat \mu - \mat X \hmat \beta}$.
The robustness of the M-scale depends on two components: (i)~the choice of the $\rho$ function and (ii)~the fixed quantity $\delta$.

The $\rho$ function in the definition of the M-scale in~\eqref{eqn:s-loss} measures the ``size'' of the standardized residuals $y_i - \hat y_i$.
The classical sample variance, up to a scaling by $\delta$, can be obtained by setting $\rho(t) = t^2$.
To get a robust estimate of scale, the $\rho$ function must be bounded \parencite{Yohai1987}, i.e., all standardized residuals larger than a certain threshold are all assigned the same ``size''.
Since the objective is to get a robust scale estimate, from here on we always assume that the $\rho$ function satisfies the condition
\begin{enumerate}[label={[A\arabic*]}, series=assumptions]
\item $\rho\colon \mathbb R \to [0, 1]$ is an even and twice continuously differentiable function with $\rho(0) = 0$, that is bounded, $\rho(t) = 1$ for all $|t| \geq c > 0$, and non-decreasing in $|t|$. \label{ass:rho-function}
\end{enumerate}
The boundedness implies that the derivative of the $\rho$ function is 0 for $|t| \geq c$.
Therefore, residuals greater than $c$ have no effect on the minimization of~\eqref{eqn:s-loss}.
As becomes evident in Theorem~\eqref{thm:asymptotic-properties} for adaptive PENSE and as shown in \textcite{Davies1990} for the unregularized S-estimator, the choice of the $\rho$ function directly affects the variance of the estimator.
\textcite{Hossjer1992} derives an ``optimal'' $\rho$ function for the unregularized S-estimator, in the sense that it minimizes the asymptotic variance for Normal errors.
However, the author also shows that the gain in efficiency is minor for Normal errors when compared to the simpler Tukey's bisquare $\rho$ function given by
\begin{equation}
\label{def:tukey-bisquare}
\rho(t; c) = \begin{cases}
	1 - \left(1 - \left( \frac{t}{c} \right)^2 \right)^3 & |t| \leq c \\
	1 & |t| > c
\end{cases}.
\end{equation}
It should be noted that the cutoff $c$ for Tukey's bisquare function does not affect the resulting S-estimator or the variance of the M-scale estimator; it is merely a multiplicative factor for the scale estimate and does not change the estimate of the regression parameters (this is true for any $\rho$ function with cutoff $c$ satisfying $\rho(t; c) = \rho(t/c; 1)$).
We are therefore fixing $c=1$ for the reminder of this paper when referring to the S-loss.
If an M-scale estimate of the scale of the residuals is desired, however, we use a cutoff $c$ which leads to a consistent estimate under Normal errors.
This cutoff will depend on $\delta$.

The second component that determines the robustness of an S-estimator is the constant $\delta$ which must be in $(0, 0.5]$ for $\rho$ functions of the form~\ref{ass:rho-function}.
The M-scale estimate can tolerate up to $\lfloor n \min(\delta, 1 - \delta) \rfloor$ gross outliers without exploding to infinity or imploding to 0 \parencite{Maronna2019}.
Theorem~\ref{thm:fbp} shows that adaptive PENSE can also tolerate up to $\lfloor n \min(\delta, 1 - \delta) \rfloor$ adversely contaminated observations without giving aberrant results.
For robustness considerations, an optimal choice is therefore $\delta=0.5$, which would allow the estimator to tolerate gross outliers in the residuals of almost 50\% of observations in the sample.
On the other hand, the variance of the estimator increases with $\delta$ and adaptive PENSE with $\delta=0.5$ achieves only \textasciitilde30\% efficiency under the Normal model while for $\delta=0.25$ the efficiency is close to 80\%.
This highlights that a good sense of the expected proportion of contaminated observations is important to get as much efficiency as possible.

The unregularized S-estimator cannot be computed if $p > n (1 - \delta) - 1$ and it cannot recover the  set of relevant predictors.
In \textcite{CohenFreue2019}, the S-loss is combined with the elastic net penalty, a generalization of the LASSO and Ridge penalties.
The elastic net penalty, $\Phi_\subt{EN}$, is a convex combination of the $L_1$ and the squared $L_2$ norm given by
\begin{equation}\label{def:en-penalty}
\Phi_\subt{EN}(\mat \beta; \lambdap, \alphap) = \lambdap \sum_{j = 1}^p \frac{1 - \alphap}{2} \beta_j^2 + \alphap \left| \beta_j \right|.
\end{equation}
The hyper-parameter $\alphap \in [0, 1]$ controls the balance between the $L_1$ and the $L_2$ penalty and $\lambdap$ controls the strength of the penalization.
The Ridge penalty is recovered when setting $\alphap= 0$, although it does not lead to variable selection.
For $\alphap = 1$, the EN penalty coincides with the LASSO, but if $\alphap < 1$, the elastic net results in a more stable variable selection than the LASSO penalty when predictors are correlated \parencite{Zou2005}.

The elastic net penalty, like the LASSO, introduces non-negligible bias and thus cannot lead to a variable selection consistent estimator.
We are therefore proposing to combine the robust S-loss with the following adaptive elastic net penalty, a slight variation of the penalty introduced by \textcite{Zou2009}:
\begin{equation}\label{def:adaptive-en-penalty}
\Phi_\subt{AE}(\mat \beta; \lambdaadap, \alphaadap, \zeta, \tmat \beta) =
	\lambdaadap \sum_{j = 1}^p \left| \tilde\beta_j \right|^{-\zeta} \left(
	   \frac{1 - \alphaadap}{2} \beta_j^2 + \alphaadap \left| \beta_j \right| \right),
\quad\quad \zeta \geq 1.
\end{equation}
The adaptive EN combines the advantages of the adaptive LASSO penalty \parencite{Zou2006} and the elastic net penalty \parencite{Zou2009}.
Contrary to the original definition in \textcite{Zou2009}, \eqref{def:adaptive-en-penalty} applies the penalty loadings $ | \tilde\beta_j |^{-\zeta}$ to both the $L_1$ and $L_2$ penalties.
The adaptive EN leverages information from a preliminary regression estimate, $\tmat \beta$, to penalize predictors with initially ``small'' coefficient values more heavily than predictors with initially ``large'' coefficients.
This has two major advantages over the non-adaptive EN penalty: (i)~the bias for large coefficients is reduced and (ii)~variable selection is improved by reducing the number of false positives.
Compared to the adaptive LASSO, the adaptive EN furthermore improves the stability of the estimator in the presence of multicollinearity \parencite{Zou2009}.

Adaptive PENSE is a two-step procedure leveraging a PENSE estimate with $\alphap = 0$, i.e., using a Ridge penalty.
In the first step, a PENSE-Ridge estimate is computed as
\begin{equation} \label{def:pense}
\tmat\theta  = \argmin_{\mu, \mat\beta}
\mathcal{O}_\subt{S}\left( \mat y, \mu + \mat X \mat\beta \right) + \Phi_\subt{EN}(\mat \beta; \lambdap, 0).
\end{equation}

In the second step, the PENSE-Ridge estimate is used as the preliminary estimate and adaptive PENSE is computed as
\begin{equation} \label{def:adaptive-pense}
\hmat\theta  = \argmin_{\mu, \mat\beta}
\mathcal{O}_\subt{S} \left( \mat y, \mu + \mat X \mat\beta \right) + \Phi_\subt{AE}(\mat \beta; \lambdaadap, \alphaadap, \zeta, \tmat \beta).
\end{equation}

Fixing the preliminary estimate to a PENSE-Ridge has two important advantages: (i)~computation is fast because the Ridge penalty is smooth (hence amenable to more efficient algorithms) and because we do not need to choose from several $\alphap$ values, and (ii)~no predictors are discarded prematurely.
While discarding some predictors in the preliminary stage may be computationally beneficial for very-high dimensional problems, empirical studies suggest variable selection performance of adaptive PENSE is better in most scenarios if the preliminary stage does not perform variable selection.

Even with the first-stage penalty fixed at $\alphap = 0$, computing adaptive PENSE estimates involves choosing several hyper-parameters: (i)~$\alphaadap$, the balance of $L_1$/$L_2$ regularization for adaptive PENSE, (ii)~$\lambdap$, the level of regularization for PENSE, (iii)~$\lambdaadap$, the level of penalization for adaptive PENSE, and (iv)~$\zeta$, the exponent in the predictor-specific regularization.
Interpreting the exponent $\zeta$ is less intuitive than the other regularization hyper-parameters.
In general, the larger $\zeta$ the stricter the differentiation between ``small'' and ``large'' coefficient values.
In other words, if $\zeta$ is large, all but a few predictors with initially very large coefficient values will be heavily penalized and thus likely not included in the set of relevant predictors.

In addition to the large number of hyper-parameters that need to be selected, both optimization problems~\eqref{def:pense} and~\eqref{def:adaptive-pense} are highly non-convex in $\mu$ and $\mat\beta$.
Finding global minima through numeric optimization is therefore contingent on a starting value that is close to a global minimum.
Section~\ref{sec:computing} describes a strategy for obtaining starting values used for adaptive PENSE.

\subsection{More robust variable selection}
The adaptive EN penalty brings the additional advantage of more robust variable selection properties compared to non-adaptive penalties.
So-called ``good'' leverage points in non-relevant predictors (i.e., observations with extreme values in one or more predictors with a coefficient value of 0 but without gross error in the response), as shown in Figure~\ref{fig:good-leverage-example}, can lead to an arbitrary number of false positives in robust estimates when using non-adaptive penalties.
Interestingly, these predictors are often the first to enter the model.
This is caused by a combination of how large values in predictors affect the sub-gradient of the objective function and robust scaling of the predictors.
The phenomenon is best seen from the sub-gradient of the PENSE objective function at $\mat\beta = \mat 0_p$, given by

$$
\subgrad{\mat\beta} \left\{ \mathcal{O}_\subt{S}\left( \mat y, \mu + \mat X \mat\beta \right) + \Phi_\subt{EN}(\mat \beta; \lambdap, \alphap) \right\} \evalat{\mat\beta = \mat 0_p}
 = -\frac{1}{n} \sum_{i = 1}^n w^2_i \left( y_i - \mu \right) \mat x_i + \lambdap {[-\alphap; \alphap]},
$$
where $w_i$ are determined by the S-loss evaluated at the intercept-only model.
These weights are $>0$ if and only if the residual from the intercept-only model is not too large (relative to all other residuals) and different from 0 (i.e., not fitted exactly).

Consider now that predictor $j$ is truly inactive and contains an extremely large value for observation $i$, but the residual for observation $i$ in the intercept-only model is small and non-zero, i.e., the $i$-th observation is a good leverage point.
An example of this scenario is shown in Figure~\ref{fig:good-leverage-example}.
Robust scaling of the predictor is likely not substantially shrinking this extremely large value and hence the $j$-th predictor dominates the sub-gradient at $\mat\beta = \mat 0_p$; therefore, it enters the model first.
In other words, this single aberrant value leads to the false impression that the $j$-th predictor is relevant.
However, because the leverage is caused by an extreme value in a non-relevant predictor, the estimated coefficient for this predictor is likely very small in magnitude, compared to coefficients of truly relevant predictors.
Importantly, the higher the leverage of this observation, the smaller the estimated coefficient.
This allows adaptive PENSE to screen out the wrongly included predictor, making it more robust against this form of contamination.

\begin{figure}[t]
{\centering \includegraphics[width=0.45\linewidth]{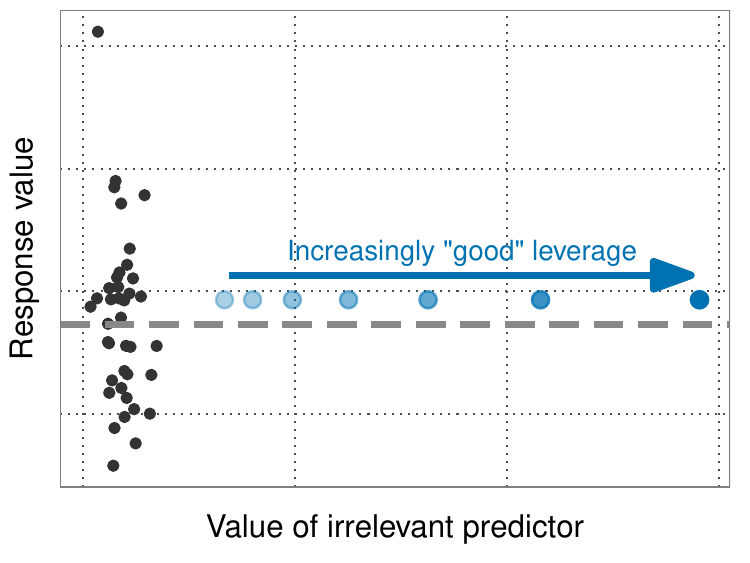}

}
\caption{%
Example of a good leverage point in a truly irrelevant predictor.
The leverage is higher as the value in the irrelevant predictor increases and the residual in the intercept-only model (depicted as dashed line) is neither 0 nor too large.}\label{fig:good-leverage-example}
\end{figure}

Interestingly, good leverage points in irrelevant predictors can -- but usually do not -- have this effect on non-robust estimators.
The prevalent scaling of the predictors using the non-robust sample standard deviation shrinks such extreme values, thereby reducing their contribution to the sub-gradient and hence their influence.
As depicted in Figure~\ref{fig:good-leverage}, however, the variable selection performance of non-robust estimators is severely damaged by bad leverage points and hence in general is unreliable under contamination.

This form of good leverage points may occur in many practical problems, particularly in very sparse settings.
In protein expression data, for example, a group of proteins could be highly expressed in a small fraction of subjects while only trace amounts of the protein are detected in the vast majority of subjects.
Even if the group of proteins is not relevant for the outcome of interest, robust methods with non-adaptive penalties are prone to selecting these proteins.
An example of this behavior using synthetic data is shown in Figure~\ref{fig:good-leverage}.
Here, 5 out of 28 irrelevant proteins have higher expression levels in 10 out of 100 observations.
Additionally, the response variable contains outliers alongside high-leverage points in 2 out of 5 relevant proteins in another 5 observations.
It is obvious that the effect of such high-leverage points is more pronounced the higher the leverage of the contaminated observations, but the estimation methods are affected differently.
Non-robust EN estimators tend to select only the 2 contaminated truly relevant predictors, but given the outlyingness of the observations, the actual parameter estimates are highly biased.
Similarly affected, I-LAMM doesn't select any proteins if the leverage caused by the affected proteins is too high.
The PENSE estimator, on the other hand, selects almost all of the affected irrelevant proteins, but as the leverage increases, tends towards not selecting any proteins at all.
Only adaptive PENSE is mostly unaffected by the contamination in relevant and irrelevant predictors, identifying on average 4 out of 5 truly relevant predictors, while screening out 24 out of 28 irrelevant predictors.

Before presenting more empirical evidence of the robustness of variable selection by adaptive PENSE in Section~\ref{sec:numerical-studies}, we discuss the intricate computational challenges and the theoretical properties of the estimator.

\begin{figure}[t]
{\centering \includegraphics[width=0.95\linewidth]{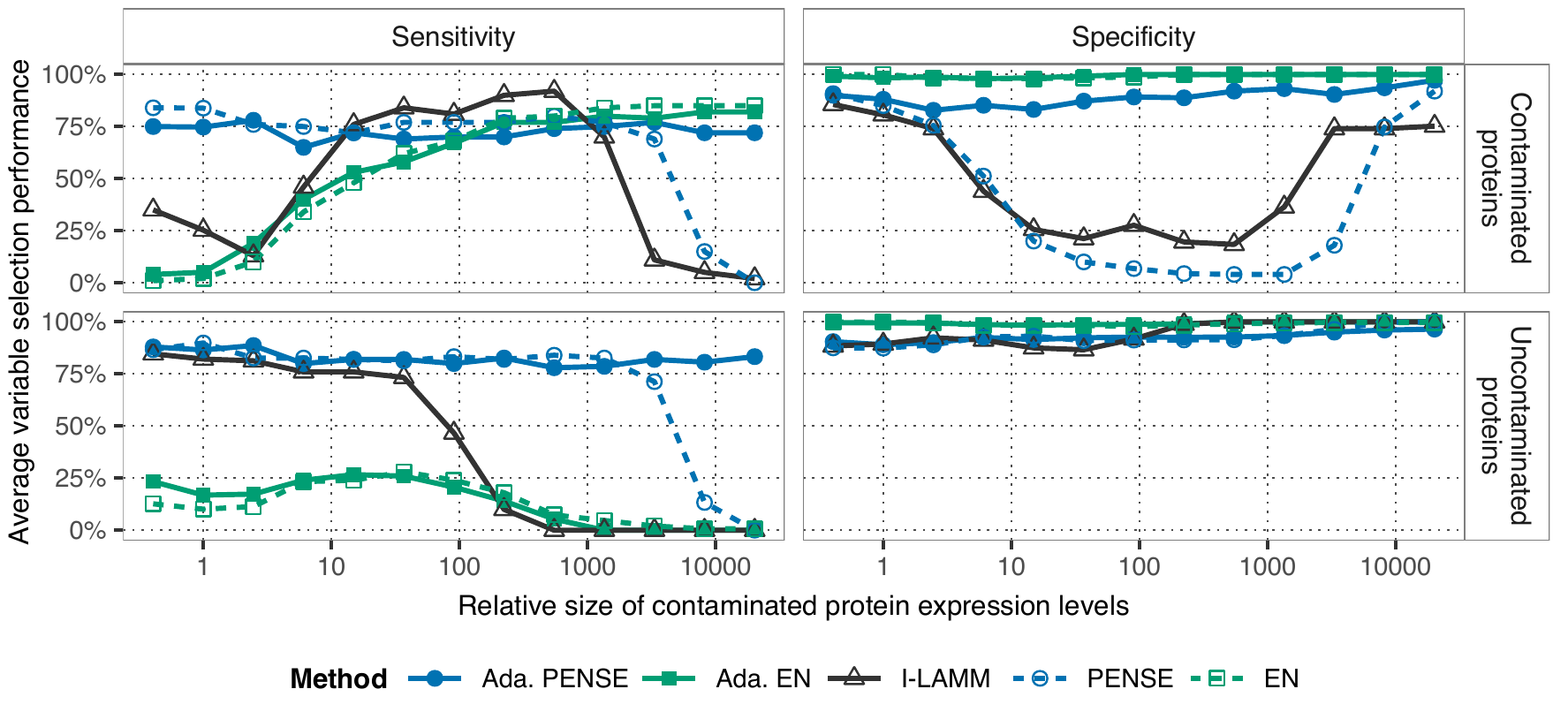}

}
\caption{%
Effect of high-leverage points on the sensitivity and specificity of various variable selection methods for synthetic data.
Average performance over 50 replications are reported separately for proteins with contaminated observations and proteins free from any contamination.
Generated data comprises $n=100$ observations of $p=32$ protein expression levels, 5 of which are relevant.
In 5\% of observations the response and 2 relevant proteins are contaminated, and in 10\% of observations 5 irrelevant proteins contain contaminated expression levels.
}\label{fig:good-leverage}
\end{figure}

\section{Computing adaptive PENSE estimates}\label{sec:computing}

The highly non-convex objective function paired with the need to select several hyper-parameters requires specialized algorithms and strategies for computing adaptive PENSE estimates.
The problem is separated into two stages: (1)~computing PENSE estimates and selecting the appropriate penalization level for the preliminary PENSE estimate and (2)~computing adaptive PENSE estimates based on the (fixed) preliminary estimate.
These two stages are done sequentially and the only information passed from stage 1 to 2 is the preliminary parameter estimate.

\subsection{Selecting hyper-parameters}\label{sec:computing-cv}

In both stages the penalization level is selected via independent repeated K-fold cross-validations (CVs) over a range of possible values.
In the first stage, only the penalization level is selected as the elastic net parameter, $\alphap$, is fixed at 0.
For the second stage, the two hyper-parameters $\alphaadap$ and $\zeta$ are selected from a small set of pairs.
For every desired pair, the penalization level is chosen via separate CVs (but using the same splits), and the combination resulting in the best prediction performance is selected.
The range of penalization levels is different in both stages, as well as for every pair of $\alphaadap$ and $\zeta$ considered in the second stage.

Estimating the prediction error of robust estimators via cross-validation suffers from high variability due to potential contamination in the data and non-convexity of the objective function.
The prediction error estimated from a single CV is highly dependent on the random split and hence unreliable for selecting the hyper-parameter.
We work around this high volatility by repeating CV several times to get a more reliable assessment of the estimator's prediction performance for given values of the hyper-parameters.
In addition to repeating CV, the measure of the prediction error needs to be stable in the presence of gross errors in the observed response values.
For PENSE and adaptive PENSE we use the highly robust $\tau$-scale of the uncentered prediction errors \parencite{Maronna2002} to estimate the prediction accuracy in each individual CV run:

$$
\hat\tau = \sqrt{
  \frac{1}{n}  \sum_{i=1}^n \min \left( c_\tau,
    \frac{ \left| y_{i} - \hat y_{i} \right| }
           { \operatorname*{Median}\limits_{i' = 1, \dotsc, n} \left| y_{i'} - \hat y_{i'} \right| }
  \right)^2
}.
$$
Here, $\hat y_{i}$ is the predicted response value from the CV split where the $i$-th observation is in the test set.
The parameter $c_\tau > 0$ specifies what constitutes outlying values in terms of multiples of the median absolute deviation and hence governs the tradeoff between efficiency and robustness of the scale estimate.

Repeated CV leads to several estimates of the prediction accuracy.
Since the presence of gross errors in the predictions is already handled by the robust $\tau$-scale, we average the prediction errors using the sample mean.
This gives an overall measure of prediction performance for a fixed set of hyper-parameters.
Repeating cross-validation furthermore gives insights into the variability of the prediction performance and affords more sensible selection of the penalization level, e.g., using the ``one-standard-error rule'' \parencite{Hastie2009}.

Another important remedy to reduce the variability incurred by CV is scaling the input data to make penalization levels more comparable across CV splits.
We first standardize the original data set by centering the response and each predictor using univariate M-estimates of location.
Then we scale each predictor to have unit M-scale and refer to this data set as the standardized input data.
In each CV split, the training data is re-standardized in the same way as the original data set.
Therefore, a fixed penalization level $\lambda$ induces a level of sparsity to the parameter estimate computed on the training data comparable to the sparsity when computed on the standardized input data.

\subsection{Algorithms for adaptive PENSE}

The algorithm to compute (adaptive) PENSE estimates is optimized for computing estimates over a fine grid of penalization levels.
In each individual CV run, hyper-parameters $\alphaadap$ and $\zeta$ are fixed.
The biggest challenge when computing adaptive PENSE estimates is the non-convexity of the objective function.
Many local minima of the objective function, however, are artifacts of contaminated observations and undesirable.
This insight is used in \textcite{CohenFreue2019} to find ``initial estimates'' for PENSE, i.e., approximate solutions which are closer to ``good'' local minima than to local minima caused by contamination.
We adapt their elastic net Peña-Yohai (EN-PY) procedure for the adaptive EN penalty to compute initial estimates for adaptive PENSE.
We call our procedure adaptive EN-PY.

Additionally we improve upon the ``warm-start'' heuristics described in \textcite{CohenFreue2019} to substantially increase exploration of the search space while maintaining computational feasibility.
For a small subset of the grid of penalization levels (e.g., every tenth value), we compute a set of initial estimates using the adaptive EN-PY procedure.
Even for a fixed penalty level, the effect of penalization on adaptive EN-PY may be vastly different than on adaptive PENSE.
Therefore, we collect all of these initial estimates (i.e., from all penalization levels) into one set of initial estimates.

Beginning at the largest value in the fine grid of penalization levels, we use every approximate solution in the set of initial estimates to start an iterative algorithm following the minimization by majorization (M-M) paradigm \parencite{Lange2016}.
The iterative M-M algorithm locates a local optimum by solving a sequence of weighted adaptive LS-EN problems, each with updated observation weights derived from the robust S-loss function evaluated at the current iterate.
Instead of fully iterating until convergence for every starting point, the M-M algorithm is stopped prematurely, leading to a set of candidates solutions.
Of those, only the most promising candidates (i.e., those with the lowest value of the objective function) are fully iterated.
The final estimate at the largest penalization level is then the fully iterated solution with smallest value of the objective function.
This two-step approach -- exploration and improvement -- is successfully applied for many other types of robust estimators \parencite[e.g.,][]{Salibian-Barrera2006,Rousseeuw2006,Alfons2013} and works very well for adaptive PENSE, too.

At the next smallest penalization level, the M-M algorithm is started from all initial estimates plus all fully iterated, most promising candidates from the previous penalization level.
Similar to the previous penalization level, only a few iterations of the M-M algorithm are performed for these starting points to reduce computation time and only the most promising solutions are iterated until convergence.
This cycle is repeated for every value in the grid of penalization levels, from largest to smallest.
Carrying forward the most promising solution from previous penalization levels combined with initial estimates from adaptive EN-PY leads to efficient and extensive exploration of the parameter space.

The computational solutions discussed here are readily available in the R package \texttt{pense}, available on CRAN (\url{https://cran.r-project.org/package=pense}).

\section{Asymptotic theory}\label{sec:theory}
To establish theoretical guarantees for adaptive PENSE, we formalize the model \eqref{def:linear-regression-model}.
We assume that the random predictors $\randX$ with distribution function $G_0$ are independent of the error term $\randU$ which follows the distribution $F_0$.
The joint distribution $H_0$ is assumed to satisfy
\begin{equation}\label{ass:joint-distribution}
H_0(\randY, \randX) = G_0(\randX) F_0(\randY - \mu^0 - \randX\tr \mat\beta^0).
\end{equation}
In the remainder of this section we omit the intercept term to make the statements more concise and easier to follow.
All of the following statements regarding the slope also apply to the intercept term.

The robustness properties of PENSE and adaptive PENSE are tightly connected to the bounded $\rho$ function and we therefore assume~\eqref{ass:rho-function}.
To establish statistical guarantees for adaptive PENSE we additionally require that the derivative of the $\rho$ function, denoted by $\psi$, satisfies
\begin{enumerate}[label={[A\arabic*]}, resume*=assumptions]
\item $t \psi(t)$, is unimodal in $|t|$.
Specifically, there exists a $c'$ with $0 < c' < c$ such that $t \psi(t)$ is strictly increasing for $0 < t < c'$ and strictly decreasing for $c' < t < c$. \label{ass:rho-function-psiunimodal}
\end{enumerate}

Assumption \ref{ass:rho-function} is standard in the robust literature for redescending M-estimators and is identical to the definition of a bounded $\rho$ function in \textcite{Maronna2019}.
Assumption \ref{ass:rho-function-psiunimodal} is a slight variation of more common assumptions on the mapping $t \mapsto t \psi(t)$, but it is nevertheless satisfied by most bounded $\rho$ functions used in robust estimation, including Tukey's bisquare function.

Finally, to establish the root-n consistency of the PENSE estimator and the oracle property of the adaptive PENSE estimator we need the same regularity conditions as in \textcite{Smucler2017}:
\begin{enumerate}[label={[A\arabic*]}, resume*=assumptions]
\item $\mathbb{P}(\mat x\tr \mat\beta = 0) < 1 - \delta$ for all non-zero $\mat\beta \in \Rm{p}$ and $\delta$ as defined in \eqref{eqn:s-loss}.\label{ass:reg-1}
\item The distribution $F_0$ of the residuals $\randU$ has an even density $f_0(u)$ which is monotone decreasing in $|u|$ and strictly decreasing in a neighborhood of 0.\label{ass:reg-2}
\item The second moment of $G_0$ is finite and $\EV[G_0]{\randX \randX\tr}$ is non-singular.\label{ass:reg-3}
\end{enumerate}
It is noteworthy that the assumption on the residuals \ref{ass:reg-2} does not impose any moment conditions on the distribution, which makes our results applicable to extremely heavy tailed errors.
Furthermore, unlike many results concerning regularized M-estimators, we only require a finite second moment of the predictors.

Under these assumptions we first establish the finite-sample robustness of adaptive PENSE.
We quantify the finite-sample robustness by the replacement finite-sample breakdown point (FBP) of an estimator $\hmat\theta$ given the sample $\mathcal Z$; the FBP $\epsilon^*(\hmat\theta; \mathcal Z)$ is defined as
\begin{equation}\label{def:fbp}
	\epsilon^*(\hmat\theta; \mathcal Z) = \max \left\{
		\frac{m}{n}: \sup_{\tilde{\mathcal{Z}} \in \mathfrak{Z}_m}  \left\| \hmat\theta(\tilde{\mathcal{Z}}) \right\| < \infty \right\},
\end{equation}
where the set $\mathfrak{Z}_m$ contains all possible samples $\tilde{\mathcal{Z}}$ with $0 \leq m < n$ of the original $n$ observations in $\mathcal{Z}$ replaced by arbitrary values \parencite{Donoho1982}.
As noted by several authors \parencite[e.g.][]{Davies2005,Smucler2017} the FBP might not be the most adequate measure of the robustness for penalized regression methods and does not address how robust the variable selection is.
However, the FBP still facilitates comparison between robust regularized regression estimators and is essential for understanding the upper limit of contamination an estimator can tolerate.

\begin{theorem}\label{thm:fbp}
For a sample $\mathcal Z = \{ (y_i, \mat x_i)\colon i = 0, \dotsc, n \}$ of size $n$, let $m(\delta) \in \mathbb{N}$ be the largest integer smaller than $n \min (\delta, 1 - \delta)$, where $\delta$ is as defined in \eqref{eqn:s-loss}.
Then, for any fixed $\lambdaadap > 0$ and $\alphaadap$, the adaptive PENSE estimator, $\hmat\theta$, retains the breakdown point of the preliminary PENSE estimator, $\tmat\theta$:
\begin{displaymath}
	\frac{m(\delta)} {n} \leq \epsilon^*\left( \hmat\theta; \mathcal Z \right) = \epsilon^*\left( \tmat\theta; \mathcal Z \right) \leq \delta \, .
\end{displaymath}
\end{theorem}

Noting that $|\tilde\beta_j|^{-\zeta} > 0$ for all $j=1,\dotsc,p$, the proof of this theorem is essentially the same as for PENSE which can be found in \textcite{CohenFreue2019}.
The main message from Theorem~\ref{thm:fbp} is that the adaptive PENSE estimator is bounded away from the boundary of the parameter space, as long as contamination is restrained to fewer than $m(\delta)$ observations.
It does not, however, mean that the estimated parameter is close to the true parameter under contamination.
Without assumption on the type of contamination, the performance of the estimator can only be assessed with numerical experiments and we show some results in Section~\ref{sec:simulation}.
If the regularization parameters are chosen by a data-driven strategy, it is important to stress that the strategy itself must provision for contamination to not break the robustness of adaptive PENSE.
In the strategy proposed in Section~\ref{sec:computing-cv}, this is ensured by repeating CV several times and using the robust $\tau$-scale for measuring the prediction accuracy.

We now turn to the asymptotic properties of adaptive PENSE as $n \to \infty$ and the dimensionality $p$ remains fixed.
Propositions~\ref{prop:strong-consistency} and \ref{prop:root-n-consistency} in the supplementary materials establish strong consistency and root-n consistency of both PENSE and adaptive PENSE estimators.
In our definition of adaptive PENSE in~\eqref{def:adaptive-pense} we take a PENSE estimate to define the penalty loadings.
With the root-n consistency of the PENSE estimator, variable selection consistency and a limiting distribution of the adaptive PENSE estimator can be derived.
These properties hold if using the PENSE-Ridge as preliminary estimate, but also for any other $\alphap$ in the computation of the preliminary PENSE estimate.

\begin{theorem}
\label{thm:asymptotic-properties}
Let $(y_i, \mat x_i\tr)$, $i = 1, \dotsc, n$, be i.i.d.\ observations with distribution $H_0$ satisfying \eqref{ass:joint-distribution}.
Under assumptions \ref{ass:rho-function}--\ref{ass:reg-3}, and if $\lambdap_n = O(1/\sqrt{n})$, $\lambdaadap_n = O(1/\sqrt{n})$, $\lambdaadap_n n^{\zeta/2} \to \infty$, and $\alphaadap > 0$, then the adaptive PENSE estimator, $\badapense$, defined in \eqref{def:adaptive-pense}  has the following properties:
\begin{enumerate}[label=(\roman*)]
\item\label{thm:var-sel-consistency}
variable selection consistency, i.e., the estimator of the truly non-relevant coefficients $\badapense_\subRnum{2}$ is zero with high probability:
          \begin{equation}
          \mathbb{P}\left( \badapense_\subRnum{2} = \mat 0_{p - s} \right) \to 1
          \quad\text{for }
          n \to \infty;
          \end{equation}

\item\label{thm:asymptotic-normality}
if in addition $\sqrt{n} \lambdaadap_n \to 0$, the estimator of the truly relevant coefficients $\badapense_\subRnum{1}$ is asymptotically Normal:
    \begin{equation}
    \label{eqn:asymptotic-normality}
    \sqrt{n} \left(\badapense_\subRnum{1} - \mat\beta^0_\subRnum{1} \right) \xrightarrow{\text{\ d\ }} N_s\left(\mat 0_s, \omscalebsq \frac{a(\psi, F_0)}{b(\psi, F_0)^2} \mat \Sigma_{\mat x_\subRnum{1}}^{-1} \right)
    \quad\text{for }
    n \to \infty.
    \end{equation}
   The constants are given by $a(\psi, F_0) = \EV[F_0]{ \psi \left( u / \omscaleb \right)^2 }$, $b(\psi, F_0) = \EV[F_0]{ \psi' \left( u / \omscaleb \right) }$, and $\mat \Sigma_{\mat x_\subRnum{1}} = \EV{\mat x_\subRnum{1} \mat x_\subRnum{1}\tr}$, i.e., the covariance matrix of the truly relevant predictors, $\mat x_\subRnum{1}$.
\end{enumerate}
\end{theorem}

The proof of Theorem~\ref{thm:asymptotic-properties} is given in the supplementary materials.
The results of Theorem~\ref{thm:asymptotic-properties} show that adaptive PENSE has the same asymptotic properties as if the true model would be known in advance, under fairly mild conditions on the distribution of the predictors and the error term.
Similar results are obtained for the MM-LASSO in \textcite{Smucler2017}, but these results depend on a good estimate of the residual scale.
What distinguishes our results from previous work is that the oracle property for adaptive PENSE can be obtained without prior knowledge of the residual scale, even under very heavy tailed errors.

The theoretical results depend on an appropriate choice for the regularization parameters $\lambdap$ and $\lambdaadap$, for PENSE and adaptive PENSE, respectively.
In practice, the choice is not obvious and the regularization parameters are usually determined through a data-driven procedure.
Depending on the procedure, however, the required conditions are difficult if not impossible to verify.
To substantiate the theoretical properties and verify that they translate beneficially to practical problems, we conduct a numerical study and carefully investigate adaptive PENSE's properties in finite samples with data-driven hyper-parameter selection as detailed in Section~\ref{sec:computing}.

\section{Numerical studies}\label{sec:numerical-studies}

The following simulation study covers scenarios where adaptive PENSE could theoretically possess the oracle property (if not for data-driven hyper-parameter selection), but also situations where the required assumptions for the oracle property are not met (e.g., when $p > n$).
Similarly, the real-world application showcases the potential of the highly robust adaptive PENSE estimator and the proposed hyper-parameter search in challenging settings where other estimators may be highly affected by contamination.

\subsection{Simulation study}\label{sec:simulation}

The aim of the simulation study is to assess the prediction and model selection performance of adaptive PENSE and how well the proposed hyper-parameter selection procedure approximates to the optimum.
We assess the reliability of the estimators by considering a large number of different contamination structures and heavy-tailed error distributions.
Below we present the results from a scenario with $n=200$ observations and $p=32$ to $p=512$ predictors.
Of those, $s = \log_2(p)$ predictors have non-zero coefficient value.
We present the results for an additional scenario with $n=100$ and more severe contamination in the supplementary materials.

We consider multiple error distributions and either no contamination or contamination in 10\% of observations.
For each combination of error distribution and presence/absence of contamination, we randomly generate 50 data sets according to the following recipe.
The contamination settings described below are chosen to be most damaging to PENSE and adaptive PENSE, while being less detrimental to the other estimators.
Even under these adverse conditions adaptive PENSE performs better or similar to other robust methods, highlighting the superior reliability of adaptive PENSE.

The predictor values are randomly generated by $K = \lfloor 1 + \sqrt{p} / 2 \rfloor$ latent variables $Z_{i1}, \dotsc, Z_{iK}$ following a multivariate t-distribution with 4 degrees of freedom and pairwise correlation of $0.1$.
The predictor values are then generated as $x_{ij} = Z_{i, 1 + \lfloor (j - 1) / K) \rfloor + I[j > s]} + \xi_{ij}$, with $ \xi_{ij}$ i.i.d.~normal with standard deviation of $0.2$.
This yields $K$ groups of predictors which are highly correlated within groups (correlation about 0.96) and mildly correlated between groups.

The response value is generated according to the true linear model
$$
y_i = \sum_{j = 1}^s x_{ij} + \epsilon_i,
\quad\quad
i = 1, \dotsc, n.
$$
Therefore, the true regression coefficient is 1 for the first $s$ predictors and 0 for the remaining $p-s$ predictors, $\btrue = (1, \dotsc, 1, 0, \dotsc, 0)\tr$.
The random noise $\epsilon_i$ is i.i.d.\ following a symmetric stable distribution for which we consider several different stability parameter values $\nu$: (i)~a light-tailed Normal distribution ($\nu = 2)$, (ii)~a moderately heavy-tailed stable distribution ($\nu = 1.33$), and (iii)~a heavy-tailed Cauchy distribution ($\nu = 1$).
The random errors are scaled such that the true model explains 25\% of the variation in the response value, measured by the empirical variance (for Normal errors) or the squared $\tau$-scale (for heavy-tailed distributions).

In settings with 10\% contamination, bad leverage points are introduced to $\log_2(p)$ randomly chosen irrelevant predictors, denoted by $\mathcal{C}$.
In the first 10\% of observations ($i = 1, \dotsc, n/10$), the values in the chosen predictors are multiplied by a constant factor $k_\subt{lev} = 2$ and a variable factor which ensures the predictor values are large relative to the correlation structure and the values in the other observations.
Several different values for $k_\subt{lev}$ were investigated; adaptive PENSE seems most affected by a moderate amount of leverage ($k_\subt{lev} = 2$), while results from all other estimators are substantially worse when introducing higher leverage ($k_\subt{lev} > 4$).
This is expected as edge-cases are the most difficult to handle for robust methods, while extremely aberrant values are easier to detect and hence their influence can be reduced.

In the same 10\% of observations, the response variable is contaminated with gross outliers.
For these observations, the response is generated by the contaminated model
\begin{align*}
y_i = - \sum_{j \in \mathcal{C}} x_{ij} + \tilde\epsilon_i,
\quad
i = 1, \dotsc, n/10.
\end{align*}
where $\mathcal{C}$ are the predictors contaminated with leverage points.
The noise $\tilde\epsilon_i$ is i.i.d.~Normal with a variance such that the contaminated model explains about 91\% of the variation in the contaminated response.
Therefore, the contaminated model creates a very strong signal for a small proportion of observations.
Moreover, the $L_1$ norm of the coefficient vector in the contaminated model is similar to the $L_1$ norm of the true coefficient vector, and hence a penalty function does not guide estimators towards the true model.
As with our choices for the bad leverage points, this contamination model was chosen as it has the most severe effect on adaptive PENSE.
Other methods, including adaptive MM, are much more affected by more severe outliers, i.e., if the coefficients in the contamination model are of larger magnitude. 

In addition to bad leverage points and outliers, all settings have 20\% of observations as good leverage points by introducing large values in $(p - s) / 2$ irrelevant predictors.
In the $n/5$ non-contaminated observations with largest Mahalanobis distance, the values in the trailing $(p-s)/2$ predictors are multiplied by a constant factor, thereby increasing the leverage of the affected observations.
The effect of the multiplicative factor is most noticeable in adaptive PENSE when a moderately high amount of leverage is introduced.
As seen in the discussion about more robust variable selection by adaptive PENSE, other estimators suffer more from very high leverage points.

The hyper-parameters for all estimators are selected according to the schema depicted in Section~\ref{sec:computing-cv} using 10 replications of 5-fold CV.
For I-LAMM and non-robust EN estimators we use the mean absolute error (MAE) as a measure of prediction accuracy.
For all other estimators, we use the uncentered $\tau$-scale estimate with $c_\tau = 3$.
The hyper-parameter $\alpha$ in EN-type penalties is chosen from the values $\{ 0.5, 0.75, 1 \}$ and for adaptive EN-type penalties we allow for $\zeta$ values in $\{1, 2\}$.
For each estimator we consider 50 values of the penalization level $\lambda$, chosen automatically by the software used to compute these estimators.
To compute the I-LAMM estimate we use the R package \texttt{I-LAMM} (available at \url{https://gitlab.math.ubc.ca/dakep/ilamm}), a derivative from the original package published alongside the paper, but extended to use the MAE instead of mean squared error.
Non-robust EN-type estimators are computed with the R package \texttt{glmnet} \parencite{Friedman2010}.
PENSE, adaptive PENSE and (adaptive) MM estimates are computed with the R package \texttt{pense}.
The (adaptive) MM estimator considered in this study is similar to the MM-LASSO \parencite{Smucler2017}, but using the more general EN penalty instead of the LASSO and leveraging the scale of the residuals from the PENSE-Ridge estimate.
The hyper-parameters are selected using the ``1-SE-rule'', i.e., the hyper-parameters leading to the sparsest solution while being within 1 standard error of the estimated CV prediction accuracy of the ``best'' hyper-parameter.

\begin{figure}[t]
  {\centering \includegraphics[width=1\linewidth]{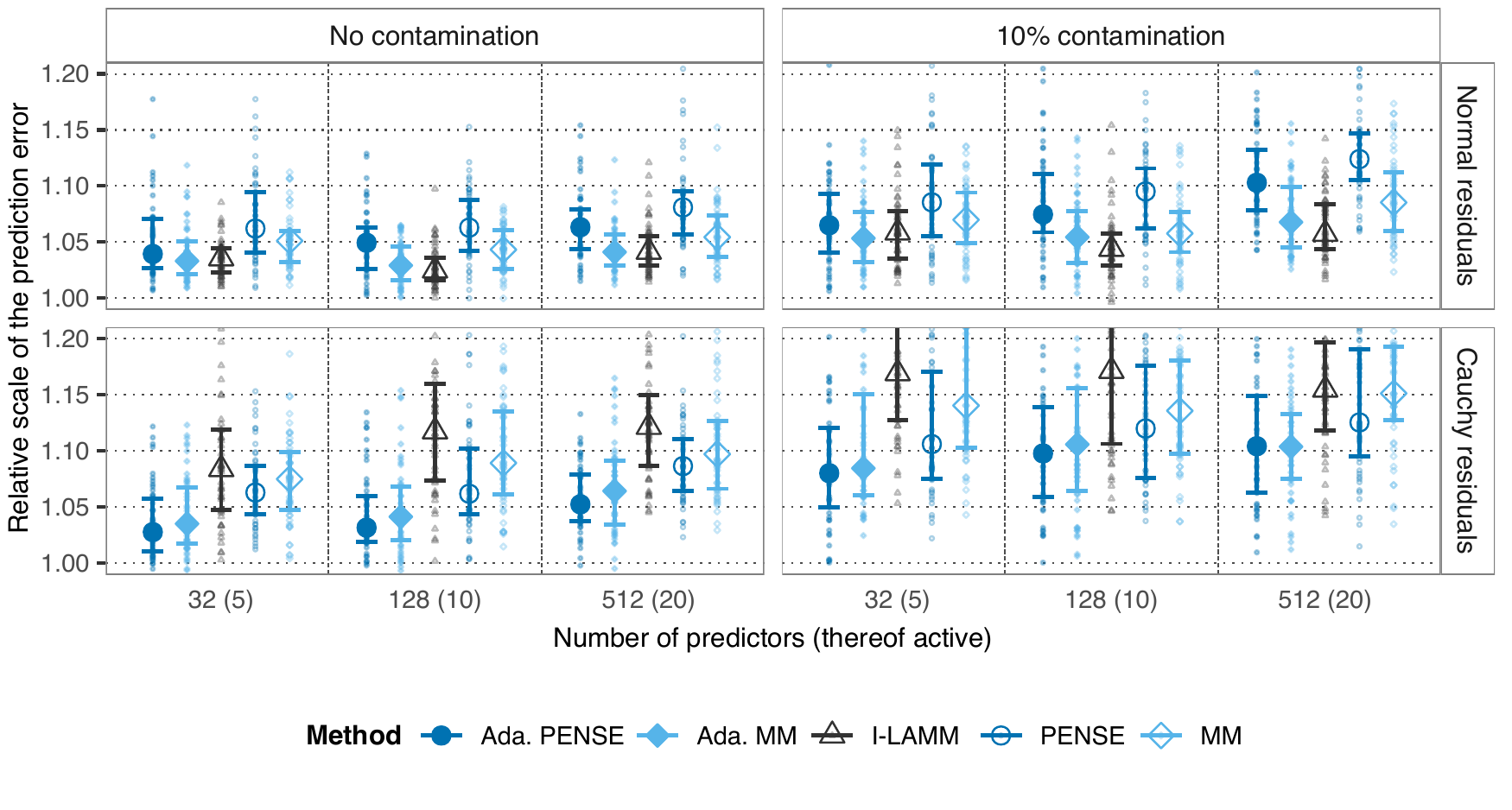}
  
  }
  \caption{%
  Prediction accuracy of robust estimators, measured by the uncentered $\tau$-scale of the prediction errors relative to the true $\tau$-scale of the error distribution (lower is better).
  The median out of 50 replications is depicted by the points and the lines show the interquartile range.
  }\label{fig:simstudy-scenario_1-predperf}
\end{figure}

The prediction performance is shown in Figure~\ref{fig:simstudy-scenario_1-predperf}.
Adaptive PENSE is less affected by contamination or heavy tails in the error distribution than other estimators.
Only adaptive MM leads to comparable performance in most settings.
Moreover, estimators with adaptive EN penalties are outperforming non-adaptive penalties due to the presence of high leverage points caused by irrelevant predictors.
With normally distributed (and moderately-heavy-tailed errors as shown in the supplementary materials), adaptive PENSE has in general lower prediction accuracy then adaptive MM; likely because the residual scale can be estimated accurately in these circumstances.
In the presence of contamination and/or heavy-tailed errors, however, adaptive PENSE is generally more reliable than adaptive MM.

Without adverse contamination, I-LAMM performs best for light-tailed error distributions, but also with heavy-tailed error distributions I-LAMM performs as well as or even better than the (non-adaptive) PENSE and MM estimators.
However, I-LAMM has the same issues with good leverage points as PENSE and is much more affected by bad leverage points and gross errors in the residuals than highly robust estimators.
Classical EN and adaptive EN estimators are not shown in these plots as they perform very poorly for heavy-tailed error distributions.
Additional plots are provided in the supplementary materials and include the least-squares EN estimators.
It is noteworthy that the considered contamination model is not very detrimental to the prediction accuracy of EN and adaptive EN due to the relatively small magnitude of the contaminated slope coefficients.
Additional results for the moderately heavy-tailed stable distribution are presented in the supplementary materials.

\begin{figure}[t]
  {\centering \includegraphics[width=1\linewidth]{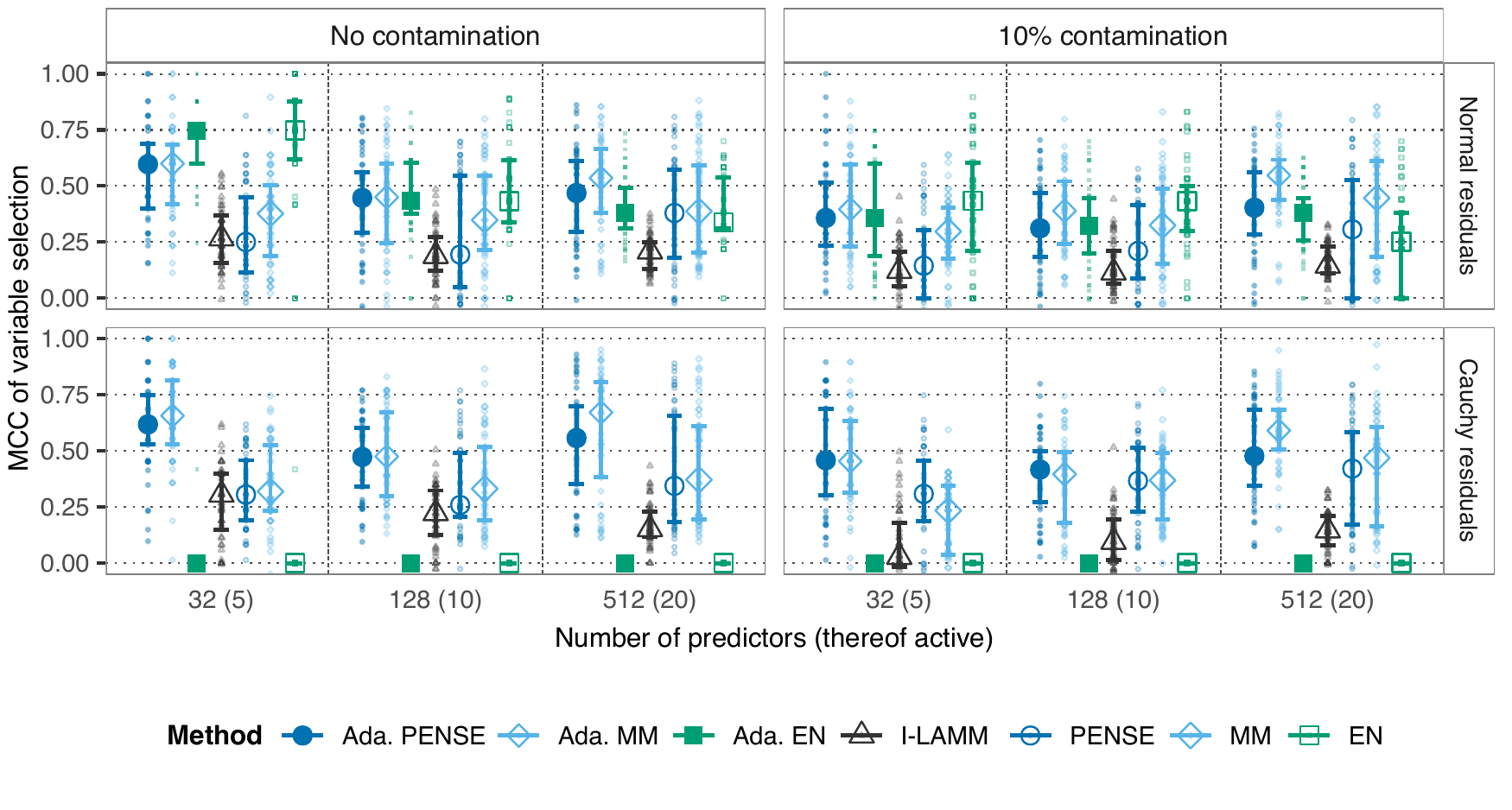}

  }
  \caption{%
  Variable selection performance of robust estimators, measured by the Matthews correlation coefficient (MCC; higher is better), defined in~\eqref{eqn:simstudy-mcc}.
  The median out of 50 replications is depicted by the points and the lines show the interquartile range.
  }\label{fig:simstudy-scenario_1-varsel}
\end{figure}

Figure~\ref{fig:simstudy-scenario_1-varsel} shows the variable selection performance of the robust estimators in the simulation study.
The depicted Matthews correlation coefficient (MCC) is calculated from the entries of the confusion matrix of true positives (TP), true negatives (TN), false positives (FP) and false negatives (FN) by
\begin{equation}\label{eqn:simstudy-mcc}
\text{MCC} =\frac{\text{TP} \cdot \text{TN} - \text{FP} \cdot \text{FN}}
 {\sqrt{(\text{TP} + \text{FP}) (\text{TP} + \text{FN}) (\text{TN} + \text{FP}) (\text{TN} + \text{FN})}}.
\end{equation}
It is evident that highly robust estimators lead to better variable selection than other estimators, especially in regimes with heavy-tailed error distributions.
Moreover adaptive penalties clearly improve variable selection upon non-adaptive penalties, and even without adverse contamination, good leverage points appear to distort variable selection for non-robust estimators.
More careful inspection of the sensitivity and specificity of variable selection (shown in the supplementary materials) underscore that non-adaptive penalties are highly affected by good leverage points and tend to select a large proportion of irrelevant predictors with large values.
While adaptive penalties screen out truly relevant predictors at a higher rate than non-adaptive penalties, adaptive PENSE and adaptive MM improve variable selection overall.
Importantly, the robust estimators with adaptive penalty are selecting none or only few of the irrelevant predictors with extreme values.
Across simulations, adaptive PENSE and adaptive MM are very similar in their variable selection performance, again noting that the contamination settings are deliberately chosen to be as detrimental as possible to adaptive PENSE and hence favor adaptive MM.
Adaptive EN has similar performance to adaptive PENSE in settings with light-tailed errors, but break down under heavy-tailed error distributions.
Similarly, I-LAMM is unable to cope with extreme values in the predictors, particularly for non-Normal errors. 

The simulation study highlights that adaptive PENSE and the strategy for choosing hyper-parameters are highly resilient towards many different forms of contamination in the predictors, even if they occur in tandem with heavy-tailed errors and gross outliers in the residuals.
Both the prediction accuracy and the variable selection performance are at least on-par with, but most often better than, other robust and non-robust regularized estimators for high-dimensional regression.
While adaptive MM is often comparable in performance to adaptive PENSE and better for Normal residuals, in certain situations it is substantially more affected by contamination.
This is particularly noticeable in some settings presented in the supplementary materials, where adaptive MM sometimes has more than 30\% higher prediction error than adaptive PENSE.

\subsection{Real-data example}
We apply adaptive PENSE and the other methods considered in the simulation study to the analysis of chemical composition of 180 archaeological glass vessels from the 15–17th century \parencite{Janssens1998}.
The analysis is performed on electron probe X-ray micro analysis spectra comprising 1920 frequencies.
This data set has been analyzed in several other papers on robust high-dimensional regression \parencite[e.g.,][]{Smucler2017,Loh2018f} as it is known the dataset contains contaminated observations both in the response variable and the frequency spectrum \parencite{Maronna2011}.
Of the 1920 frequencies available, only 487 frequencies with meaningful variation between vessels are used for the analysis, in line with the analyses conducted in comparable studies.
The goal is to predict the concentration of the chemical compound $\mathrm{P_2 O_5}$, measured as the total amount of the compound relative to the total weight of the glass fragment $[\% w/w]$.
To get a predictive model, we model the log-concentration of the chemical compound $\mathrm{P_2 O_5}$ as a linear function of the spectrum.
With similar dimensions and potential contamination as analyzed in the above simulation study, we can be confident that both adaptive PENSE and adaptive MM are very good candidates for fitting this predictive model.

The breakdown point for the robust estimators is set to 28\%, affording up to 50 observations with contaminated residuals.
Hyper-parameters are selected via 6-fold CV, repeated 10 times for all considered estimators.
The $\alpha$ parameter for all EN-type estimators is selected from the values $\{0.5, 0.75, 1\}$, and the $\zeta$ parameter for all estimators with adaptive EN penalty is chosen from $\{ 1, 2 \}$.
For non-robust EN estimators and for I-LAMM the mean absolute prediction error is used as a measure of prediction accuracy during CV.
For the robust estimators, the robust $\tau$-scale guides the hyper-parameter selection.

Figure~\ref{fig:glass-prediction-perf} shows the estimated prediction accuracy from 50 replications of 6-fold CV.
In each fold, the hyper-parameters for all estimators are chosen via an inner CV as explained in detail in Section~\ref{sec:computing-cv}.
The prediction accuracy in each CV run is estimated by the uncentered $\tau$-scale of the prediction errors for all 180 observations and every estimator (robust and non-robust).
The Figure~\ref{fig:glass-prediction-perf} shows the prediction accuracy both on the log-scale (i.e., the scale on which the predictive models are fitted) and the original scale where the predicted values are back-transformed prior to computing the prediction accuracy.

The adaptive PENSE clearly outperforms the other methods for predicting the concentration of $\mathrm{P_2 O_5}$ from the spectrum, both on the log-scale and the original scale.
The I-LAMM estimator is omitted from these plots because its prediction accuracy is substantially worse than the other methods and only marginally better than the intercept-only model, with a median $\tau$-scale of the prediction error of $0.586$ on the log scale ($0.167$ on the original scale).
It is interesting that the classical EN estimator is performing only slightly worse than some highly robust estimators.
In addition, non-robust estimators and adaptive MM show a very high amount of variation in the prediction accuracy.
This is also reflected in the number of relevant predictors estimated by these methods, which varies widely between the 50 different CV splits for some estimators.
EN, for example selects anywhere between 4 to 34 frequencies, while adaptive MM selects between 8 and 45.
While adaptive PENSE selects more frequencies than other estimators (between 42 and 49), the selection seems much more stable.
This suggests that several frequencies are contaminated by extreme or unusual values with deleterious effects on most methods except for adaptive PENSE.

\begin{figure}[t]
  {\centering 
    \subfigure[Prediction accuracy.]{\label{fig:glass-prediction-perf}%
    \includegraphics[width=0.66\linewidth]{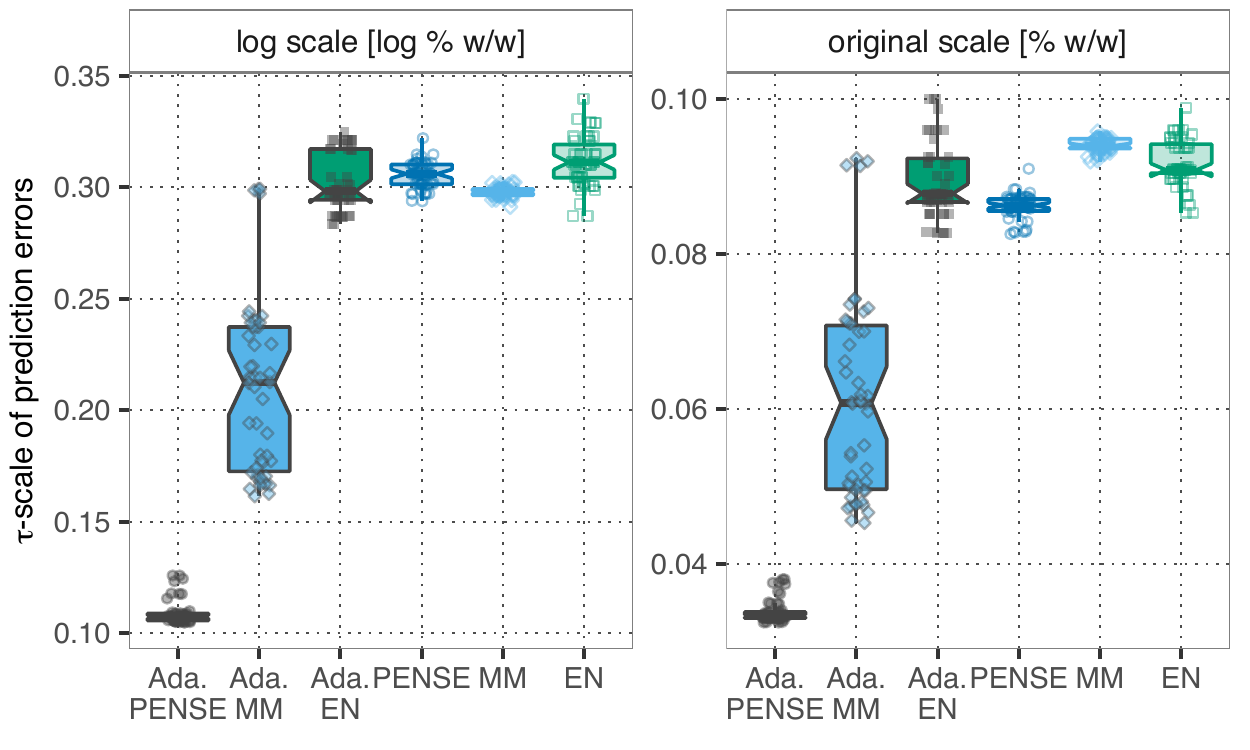}}
    \subfigure[Fitted vs. observed.]{\label{fig:glass-fitted-vs-observed}%
    \includegraphics[width=0.33\linewidth]{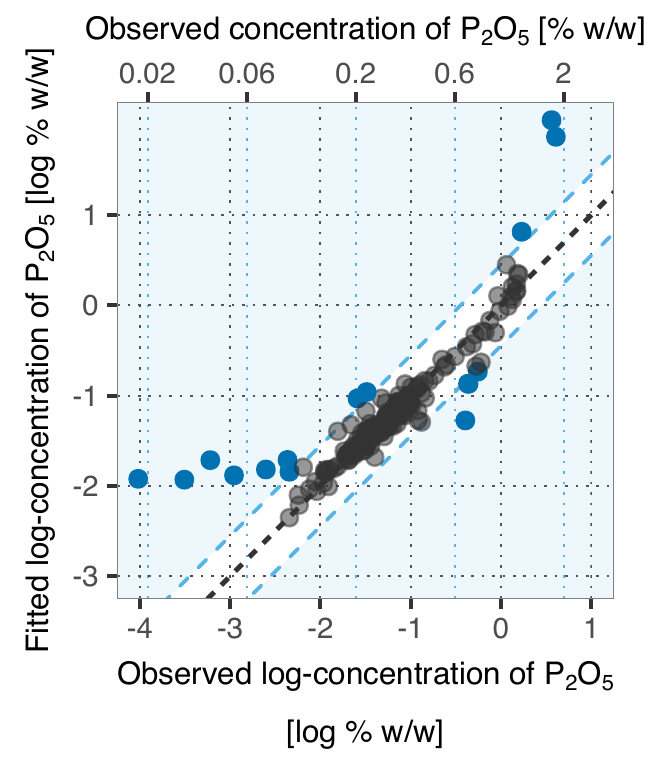}}
  
  }
  \caption{%
  Accuracy of predicting the concentration of compound $\mathrm{P_2 O_5}$ from the glass-vessel data set (a) and observed vs.\ fitted values from adaptive PENSE (b).
  For (a) the prediction accuracy is estimated by the uncentered $\tau$-scales of the prediction errors from 50 replications of 6-fold cross-validation.
  Hyper-parameters are selected independently in each of these CVs via an ``inner'' 6-fold CV using 10 replications.
  In (b) the 14 data points in the shaded areas have unusually low or high log-concentrations of $\mathrm{P_2 O_5}$ according to adaptive PENSE, with residuals larger than 3 times the estimated residual scale.
  }\label{fig:glass-application}
\end{figure}

The adaptive PENSE fit computed on all 180 glass vessels shown in Figure~\ref{fig:glass-fitted-vs-observed} suggests about 14 vessels have unusually large residuals.
The adaptive PENSE fit also suggests a moderately heavy-tailed error distribution, further explaining why adaptive PENSE is performing better than other estimators in this application.
As demonstrated in the numerical experiments above, in such a setting with an heavy-tailed error distribution the contamination in the frequency spectrum can have very detrimental effects on non-adaptive and non-robust estimators.
This is also a likely explanation why I-LAMM performs very poorly in this application and does not select any predictors.
While non-robust EN has better prediction performance than I-LAMM in this application, the variable selection seems highly affected by the contamination, with only 4 frequencies being selected.
In comparison, adaptive PENSE has been shown to retain reliable variable selection performance in such a challenging scenario and selects 43 frequencies belonging to several groups of adjacent and hence highly correlated frequencies.
Adaptive MM selects a similar but slightly larger set of 48 frequencies and these slight differences likely translate to less accurate predictions of the concentration of $\mathrm{P_2 O_5}$.

\section{Conclusions}
Unusually large or small values in predictors, paired with heavy-tailed error distributions or even gross outliers in the residuals can have severe ramifications for statistical analyses if not handled properly.
Particularly in high-dimensional problems such extreme values are highly likely.
Whether these extreme values have a detrimental effect on the statistical analysis, however, is unknown.
Omitting affected observations or predictors is therefore ill-advised and even fallacious; these anomalous values are often well hidden in the multivariate structure and thus difficult if not impossible to detect.
As a better alternative, we propose the adaptive PENSE estimator which can cope with such unusual values in the predictors even if they are paired with aberrantly large residuals.

We have demonstrated that adaptive PENSE leads to estimates with high prediction accuracy and reliable and strong variable selection even under very challenging adverse contamination settings.
Unlike other robust estimators, adaptive PENSE is capable of correctly screening out truly irrelevant predictors even if they contain aberrant and unusual values.
The extensive simulation study shows that adaptive PENSE achieves overall better prediction and variable selection performance than competing robust regularized estimators and semi- or non-robust methods such as I-LAMM \parencite{Fan2018} or least-squares (adaptive) EN.
While adaptive MM performs better for Normal errors and often similar for heavy-tailed error distributions, in some settings adaptive MM can be substantially affected by contamination.
Overall, adaptive PENSE is more resilient in challenging scenarios than other estimators considered here.
This is underscored by adaptive PENSE's superior prediction performance for predicting the concentration of the compound $\mathrm{P_2 O_5}$ in ancient glass vessels from their spectra.
Adaptive PENSE not only achieves better prediction accuracy, it does so with a more parsimonious model than other robust methods and more reliability than non-robust methods.

In addition to the strong empirical performance, we have established theoretical guarantees for the estimator.
Adaptive PENSE is asymptotically able to uncover the true set of relevant predictors and the estimator of the respective coefficients converges to a Normal distribution.
Importantly, these guarantees hold regardless of the tails of the error distribution and overall very mild assumptions on the distribution of the predictors or the residuals.

Computing adaptive PENSE estimates is challenging, but the proposed computational solutions ensure a wide range of analyses and problem sizes are amenable to adaptive PENSE.
The algorithms and the proposed hyper-parameter search are made available in the R package \texttt{pense} (\url{https://cran.r-project.org/package=pense}).
The high reliability even under adverse contamination in predictors and responses alike, combined with the readily available computational tools, make adaptive PENSE a feasible alternative in many statistical applications.

\section{Acknowledgement}
The author would like to thank Gabriela~V.\ Cohen~Freue for her feedback on early drafts of this manuscript and numerous stimulating discussions.
Ezequiel Smucler provided important input and proofreading during the development the proofs.
Keith Crank gave valuable writing suggestions to polish the final draft.
The numerical studies and the real-world application were enabled by computing resources provided by WestGrid (\url{https://www.westgrid.ca/}) and ComputeCanada (\url{https://www.computecanada.ca}), as well as by resources provided by the Office of Research Computing at George Mason University (\url{https://orc.gmu.edu}) and funded in part by grants from the National Science Foundation (Awards Number 1625039 and 2018631).

\storecounter{equation}{eqncounter}

\printbibliography

\begin{appendices}
\input{supp-proofs}

\newpage

\input{supp-sim}

\end{appendices}

\end{document}

%% file: common-cmds.tex
\usepackage{xfrac} 
\usepackage{relsize} 
\usepackage{xifthen} 
\usepackage{amsmath} 
\usepackage{amssymb} 
\usepackage{amsthm} 
\usepackage{mathrsfs} 
\usepackage{mathtools}
\usepackage{bm} 

\newcommand{\subt}[1]{%
	\mathsmaller{\text{#1}}%
}

\newcommand{\loss}[1][]{%
	\ifthenelse{\isempty{#1}}%
	{\mathcal{O}}%
	{\mathcal{O}_\subt{#1}}%
}

\newcommand{\objf}[1][]{%
	\ifthenelse{\isempty{#1}}%
	{\widetilde{\mathcal{O}}}%
	{\widetilde{\mathcal{O}}_\subt{#1}}%
}

\newcommand{\Rnum}[1]{{\uppercase\expandafter{\romannumeral #1\relax}}}

\newcommand{\subRnum}[1]{\subt{\uppercase\expandafter{\romannumeral #1\relax}}}

\newcommand{\mat}[1]{\bm{\mathrm{#1}}}

\newcommand{\tmat}[1]{%
	\ifthenelse{\equal{#1}{W}}{\bm{{\widetilde{\mathrm{#1}}}}}%
	{\bm{{\tilde{\mathrm{#1}}}}}%
}
\newcommand{\hmat}[1]{\bm{{\hat{\mathrm{#1}}}}}
\newcommand{\tr}{^\intercal}
\newcommand{\Rm}[1]{\mathbb{R}^{#1}}

\newcommand{\argmin}{\operatorname*{arg\,min}}

\newcommand{\evalat}[1]{%
  \raisebox{-2pt}{${\mathlarger{\mathlarger{\mathlarger{\mathlarger{\mid}}}}}_{#1}$}%
}

\newcommand{\subgrad}[1]{%
  \raisebox{-2pt}{$\mathlarger{\mathlarger{\mathlarger{\boldsymbol\nabla}}}\!\!_{#1}$}%
}

\makeatletter
\newcommand{\EV}{\@ifstar{\@EVinline}{\@EVdisplay}}
\newcommand{\@EVdisplay}[2][]{%
	\ifthenelse{\isempty{#1}}{\mathbb{E} \left[ #2 \right]}%
	{\mathbb{E}_{#1} \left[ #2 \right]}%
}
\newcommand{\@EVinline}[2][]{%
	\ifthenelse{\isempty{#1}}{\mathbb{E} [#2]}%
	{\mathbb{E}_{#1}[ #2]}%
}
\makeatother

\newcommand{\almostsure}{\xrightarrow{\text{a.s.}}}
\newcommand{\almostsuren}{\xrightarrow[n \to \infty]{\text{a.s.}}}

\newcommand{\randX}{\boldsymbol{\mathcal{X}}}
\newcommand{\randY}{\mathcal{Y}}
\newcommand{\randE}{\mathcal{U}}
\newcommand{\randU}{\mathcal{U}}

\newcommand{\lambdap}{{\tilde\lambda}}
\newcommand{\lambdaadap}{{\lambda}}
\newcommand{\alphap}{{\tilde\alpha}}
\newcommand{\alphaadap}{{\alpha}}

\newcommand{\trueparam}[1][]{%
	\ifthenelse{\isempty{#1}}%
	{\mat\theta^{\subt{0}}}%
	{\mat\theta^{\subt{0}}_{#1}}%
}
\newcommand{\btrue}[1][]{%
	\ifthenelse{\isempty{#1}}%
	{\mat\beta^{\subt{0}}}%
	{\mat\beta^{\subt{0}}_{#1}}%
}
\newcommand{\mutrue}[1][]{%
	\ifthenelse{\isempty{#1}}%
	{\mu^{\subt{0}}}%
	{\mu^{\subt{0}}_{#1}}%
}

\newcommand{\pense}[1][]{%
	\ifthenelse{\isempty{#1}}%
	{\tmat\theta}%
	{\tmat\theta_{#1}}%
}
\newcommand{\adapense}[1][]{%
	\ifthenelse{\isempty{#1}}%
	{\hmat\theta}%
	{\hmat\theta_{#1}}%
}

\newcommand{\bpense}[1][]{%
	\ifthenelse{\isempty{#1}}%
	{\tmat\beta}%
	{\tmat\beta_{#1}}%
}
\newcommand{\bpenseel}{\tilde\beta}
\newcommand{\mupense}[1][]{\tilde\mu}
\newcommand{\badapense}[1][]{%
	\ifthenelse{\isempty{#1}}%
	{\hmat\beta}%
	{\hmat\beta_{#1}}%
}
\newcommand{\muadapense}[1][]{\hat\mu}
\newcommand{\badapenseel}{\hat\beta}

\newcommand{\mhscale}[1]{{\hat{\sigma}_\subt{M} ( #1 )}}
\newcommand{\mhscalesq}[1]{{\hat{\sigma}^2_\subt{M} ( #1 )}}
\newcommand{\mscale}[1]{{\sigma_\subt{M} ( #1 )}}
\newcommand{\mscalesq}[1]{{\sigma^2_\subt{M} ( #1 )}}

\newcommand{\omhscalesqb}{\mhscalesq{\btrue}}
\newcommand{\omscaleb}{\mscale{\btrue}}
\newcommand{\omscalebsq}{\mscalesq{\btrue}}

\newcommand{\mhscaleast}{\mhscale{\mat \beta_n^*}}
\newcommand{\mhscaleada}{\mhscale{\badapense}}
\newcommand{\mhscaleadasq}{\mhscalesq{\badapense}}

\newtheorem{theorem}{Theorem}
\newtheorem{proposition}{Proposition}
\newtheorem{lemma}{Lemma}

%% file: supp-proofs.tex
\section{Proofs}
Below are the proofs of asymptotic properties of adaptive PENSE as presented in Section~\ref{sec:theory}.
For notational simplicity, the intercept term is dropped from the model, i.e., the linear model~\ref{def:linear-regression-model} is simplified to
$$
\randY = \randX\tr \btrue + \randE
$$
and the joint distribution $H_0$ of $(\randY, \randX)$ is written in terms of the error
$$
H_0(u, \mat x) := H_0(y, \mat x) = H_0(\mat x) F_0(y - \mat x\tr \btrue).
$$
All the proofs also hold for the model with an intercept term included.
Another notational shortcut in the following proofs is to write the M-scale of the residuals in terms of the regression coefficients, i.e.,
$\mhscale{\mat\beta}$ is shorthand for $\mhscale{\mat y - \mat X \mat\beta}$, and accordingly for the population version $\mscale{\mat\beta}$,
For all proofs below, we define $\psi(t) = \rho'(t)$ to denote the first derivative of the $\rho$ function in the definition of the M-scale estimate and hence of the S-loss, as well as the mapping $\varphi \colon \mathrm{R} \to [0; c]$ as
$$
\varphi(t) := \psi(t) t.
$$
Moreover, the complete objective functions of PENSE and adaptive PENSE are denoted by
\begin{align*}
\objf[S](\mat\beta) &= \loss[S]\left( \mat y, \mat X \mat\beta \right) + \Phi_\subt{EN}(\mat \beta; \lambdap, \alphap)\\
\text{and} \quad
\objf[AS](\mat\beta) &= \loss[S] \left( \mat y, \mat X \mat\beta \right) + \Phi_\subt{AE}(\mat \beta; \lambdaadap, \alphaadap, \zeta, \tmat \beta),
\end{align*}
respectively.

\subsection{Preliminary results concerning the M-scale estimator}
Before proving asymptotic properties of the adaptive PENSE estimator, several intermediate results concerning the M-scale estimator are required.

\begin{lemma}\label{lem:slln}
Let $(y_i, \mat x_i\tr)$, $i = 1, \dotsc, n$, be i.i.d.\ observations with distribution $H_0$ satisfying \eqref{ass:joint-distribution} and $u_i = y_i - \mat x_i\tr \btrue$.
If $\mat v \in \Rm{p}$ and $s \in (0, \infty)$ positive, then the empirical processes $\left( P_n \eta_{\mat v, s} \right)_{\mat v, s}$ with
\begin{displaymath}
\eta_{\mat v, s} (u, \mat x) := \varphi \left(\frac{u + \mat x\tr \mat v} {s} \right)
\end{displaymath}
converge uniformly almost sure:
\begin{equation}
\lim_{n \to \infty} \sup_{\substack{\mat v \in \Rm{p} \\ s \in (0, \infty)}}
	\left |
		\frac{1}{n} \sum_{i = 1}^{n} \eta_{\mat v, s} (u_i, \mat x_i) - \EV[H_0]{\eta_{\mat v, s} (\randE, \randX)}
	  \right | = 0
	  \quad \text{a.s.}
\end{equation}
\end{lemma}

\begin{proof}[Proof of Lemma~\ref{lem:slln}]
We will show step by step that the space $\mathscr{F} = \{\eta_{\mat v, s} : \mat v \in \Rm{p}, s \in (0, \infty) \}$ is a bounded Vapnik–Chervonenkis (VC) class of functions and hence Glivenko-Cantelli.
The space $\mathscr{F}$ is bounded because $\varphi(t)$ is bounded by assumptions on $\rho$.
Define the mapping
\begin{displaymath}
g_{\mat v, s} := \left\{
\begin{array}{ll}
	\Rm{p+1} & \to \mathbb{R} \\
	\left(
		\begin{array}{c}
			u \\ \mat x
		\end{array}
	\right) & \mapsto (u - \mat x\tr \mat v) s^{-1}
\end{array}
\right. .
\end{displaymath}
The corresponding function space $\mathscr{G} = \{ g_{\mat v, s} : \mat v \in \Rm{p}, s \in (0, \infty) \}$ is a subset of a finite-dimensional vector space with dimension $\operatorname{dim}(\mathscr{G}) = p + 1$.
Therefore, $\mathscr{G}$ is VC with VC index $V(\mathscr{G}) \leq p + 3$ according to Lemma~2.6.15 in \textcite{Vandervaart1996}.
Due to the assumptions on $\rho$, the function $\varphi(t)$ can be decomposed into
\begin{displaymath}
\varphi(t) = \max \{ \min\{\varphi_1(t), \varphi_2(t) \}, \min\{\varphi_1(-t), \varphi_2(-t) \} \}
\end{displaymath}
with $\varphi_{1,2}$ monotone functions.
Thus, $\Phi_{1,2} = \{ \varphi_{1,2} (g(\cdot)) : g \in \mathscr{G} \}$ and $\Phi_{1,2}^{(-)} = \{ \varphi_{1,2} (-g(\cdot)) : g \in \mathscr{G} \}$ are also VC due to Lemma~2.6.18~(iv) and (viii) in \textcite{Vandervaart1996}.
Using Lemma~2.6.18~(i) in \textcite{Vandervaart1996} then leads to  $\Phi = \Phi_1 \land \Phi_2$ and $\Phi^{(-)} = \Phi_1^{(-)} \land \Phi_2^{(-)}$ also being VC.
Finally, $\mathscr{F} = \Phi \lor \Phi^{(-)}$ is VC because of Lemma~2.6.18~(ii).
Since $\mathscr{F}$ is bounded, Theorem~2.4.3 in \textcite{Vandervaart1996} concludes the proof.
\end{proof}

\begin{lemma}\label{lem:as-convergence-vn-as}
Let $(y_i, \mat x_i\tr)$, $i = 1, \dotsc, n$, be i.i.d.~observations with distribution $H_0$ satisfying \eqref{ass:joint-distribution} and $u_i = y_i - \mat x_i\tr \btrue$.
Under assumptions \ref{ass:reg-1}, \ref{ass:reg-2} and if $\mat \beta_n^* = \btrue + \mat v_n$ with $\lim_{n \to \infty} \| \mat v_n \|  = 0$ a.s., then we have
\begin{enumerate}[label=(\alph*)]
\item \label{lem:as-convergence-vn-as-scale} almost sure convergence of the estimated M-scale to the population M-scale of the error distribution
\begin{displaymath}
	\lim_{n \to \infty} \mhscaleast \almostsure \omscaleb
\end{displaymath}
\item \label{lem:as-convergence-vn-as-denom} and almost sure convergence of
\begin{displaymath}
	\lim_{n \to \infty}
	\frac{1}{n} \sum_{i = 1}^{n} \varphi\left( \frac{u_i - \mat x_i\tr \mat v_n} {\mhscaleast} \right) =  \EV[F_0]{\varphi \left( \frac{\randE}{\omscaleb} \right) } 
	\quad \text{a.s.}
\end{displaymath}
\end{enumerate}
\end{lemma}
\begin{proof}[Proof of Lemma~\ref{lem:as-convergence-vn-as}]
The first result \ref{lem:as-convergence-vn-as-scale} is a direct consequence of the conditions of the lemma ($u - \mat x \tr \mat v_n \to u$ a.s.) and Theorem~3.1 in \textcite{Yohai1987}.

For part \ref{lem:as-convergence-vn-as-denom}, it is know from Lemma~\ref{lem:slln} the empirical process converges uniformly almost sure.
Since $\omscaleb > 0$, the continuous mapping theorem gives $\frac{u_i - \mat x_i\tr \mat v_n} {\mhscaleast} \to \frac{\randE} {\omscaleb}$ almost surely.
Finally, due to the continuity and boundedness of $\varphi$:
\begin{equation}
\EV[H_0]{\varphi \left( \frac{\randE - \randX\tr \mat v_n} {\mhscaleast} \right)}
\almostsuren
\EV[F_0]{\varphi \left( \frac{\randE}{\omscaleb} \right) }
\end{equation}
which concludes the proof.
\end{proof}

\begin{lemma}\label{lem:uniform-convergence_v_compact}
Let $(y_i, \mat x_i\tr)$, $i = 1, \dotsc, n$, be i.i.d.~observations with distribution $H_0$ satisfying \eqref{ass:joint-distribution} and $u_i = y_i - \mat x_i\tr \btrue$.
Under regularity conditions \ref{ass:reg-1}--\ref{ass:reg-3} and if $\mat v \in K \subset \Rm{p}$ with $K$ compact and $\mat \beta_n^* = \btrue + \mat v / \sqrt{n}$, then
\begin{enumerate}[label=(\alph*)]
\item the M-scale estimate converges uniformly almost sure
\begin{equation}\label{eqn:lemma3-statement-scale}
\sup_{\mat v \in K} \left| \mhscaleast - \omscaleb \right| \almostsure 0,
\end{equation}
\item for every $\epsilon > 0$ with $\epsilon < \EV[F_0]{\varphi \left( \frac{\randE}{\omscaleb} \right) }$ the uniform bound over $\mat v \in K$
\begin{equation}\label{eqn:lemma3-statement}
	\sup_{\mat v \in K} \left|
	\frac{\mhscaleast} { \frac{1}{n} \sum_{i = 1}^{n} \varphi\left( \frac{u_i - \mat x_i\tr \mat v / \sqrt{n}} {\mhscaleast} \right) }
	\right| < \frac{ \epsilon + \omscaleb } { \EV[F_0]{\varphi \left( \frac{\randE}{\omscaleb} \right) } - \epsilon }
\end{equation}
holds with arbitrarily high probability if $n$ is sufficiently large.
\end{enumerate}
\end{lemma}

\begin{proof}[Proof of Lemma~\ref{lem:uniform-convergence_v_compact}]
The proof for~\eqref{eqn:lemma3-statement-scale} relies on Lemma~4.5 from \textcite{Yohai1986} which states that under the same conditions as for this lemma, the following holds:
\begin{displaymath}
\sup_{\mat v \in K} | \mhscaleast - \mscale{\mat\beta_n^*} | \almostsure 0.
\end{displaymath}
Therefore, the missing step is to show that $\sup_{\mat v \in K} | \mscale{\mat\beta_n^*} - \omscaleb | \to 0$ almost surely as $n \to \infty$.
This is done by contradiction.

Assume there exists a subsequence $(n_k)_{k>0}$ such that for all $k$, $\sup_{\mat v \in K} | \mscale{\mat\beta_n^*} - \omscaleb | > \epsilon > 0$. Since $\mat v \in K$ with $K$ a compact set, for every sequence $\mat v_n$ there exists a subsequence $(\mat v_{n_k})_k$ such that $| \mscale{\btrue + \mat v_{n_k} / \sqrt{n_k}} - \omscaleb | > \epsilon$ for all $n_k > N_\epsilon$.
Therefore, either one of the following holds: (i)~$\mscale{\btrue + \mat v_{n_k} / \sqrt{n_k}} > \omscaleb + \epsilon$ or (ii)~$\mscale{\btrue + \mat v_{n_k} / \sqrt{n_k}} < \omscaleb - \epsilon$.
In the first case (i) it is know that
\begin{align*}
\rho\left( \frac{\randE - \randX \tr \mat v_{n_k} / \sqrt{n}} {\mscale{\btrue + \mat v_{n_k} / \sqrt{n_k}}} \right) <
	\rho\left( \frac{\randE - \randX\tr \mat v_{n_k} / \sqrt{n}} {\omscaleb + \epsilon} \right)
	\to \rho\left( \frac{\randE} {\omscaleb + \epsilon} \right).
\end{align*}
Due to the boundedness of $\rho$, the dominated convergence theorem gives
\begin{align*}
\EV[H_0]{\rho\left( \frac{\randE - \randX \tr \mat v_{n_k} / \sqrt{n}} {\mscale{\btrue + \mat v_{n_k} / \sqrt{n_k}}} \right)} <
	\EV[H_0]{\rho\left( \frac{\randE - \randX \tr \mat v_{n_k} / \sqrt{n}} {\omscaleb + \epsilon} \right)}
	\to \EV[H_0]{\rho\left( \frac{\randE} {\omscaleb + \epsilon} \right)} < \delta
\end{align*}
which contradicts the definition of $\mscale{\btrue + \mat v_{n_k} / \sqrt{n_k}}$.
In case~(ii) similar steps yield
$$
\EV[H_0]{\rho\left( \frac{\randE - \randX \tr \mat v_{n_k} / \sqrt{n}} {\mscale{\btrue + \mat v_{n_k} / \sqrt{n_k}}} \right)} > \delta
$$
for all $n_k > N$ with $N$ large enough.
Therefore, the assumption $\sup_{\mat v \in K} | \mscale{\mat\beta_n^*} - \omscaleb | > \epsilon > 0$ can not be valid and hence $\sup_{\mat v \in K} | \mscale{\mat\beta_n^*} - \omscaleb | \to 0$.
This concludes the proof of \eqref{eqn:lemma3-statement-scale}.

Before proving \eqref{eqn:lemma3-statement}, note that $\epsilon$ is well defined because $\EV[F_0]{\varphi \left( \frac{\randE}{\omscaleb} \right) } > 0$ as per Lemma~6 in \textcite{Smucler2018}.
To prove \eqref{eqn:lemma3-statement}, we first bound the denominator uniformly over $\mat v \in K$.
From Lemma~\ref{lem:slln} it is known that the empirical processes converge almost surely, uniformly over $\mat v \in K$ and $s > 0$.
As a next step, we show the deterministic uniform convergence of
\begin{equation}\label{eqn:lemma3-uniform-conv}
\sup_{\substack{\mat v \in K \\ s \in [\omscaleb - \epsilon_1, \omscaleb + \epsilon_1]}} \left|
		\EV[H_0]{f_n(\randE, \randX, \mat v, s) } -
		\EV[H_0]{\varphi \left( \frac{\randE}{s} \right) }
	\right| \to 0,
\end{equation}
where $f_n(\randE, \randX, \mat v, s)$ is defined as
\begin{displaymath}
f_n(\randE, \randX, \mat v, s) := \varphi \left( \frac{\randE -  \randX\tr \mat v / \sqrt{n}}{s} \right).
\end{displaymath}
The functions $f_n(\randE, \randX, \mat v, s)$ are bounded and converge pointwise to $\varphi \left( \frac{\randE}{s} \right)$, entailing pointwise convergence of $\EV[H_0]{f_n(\randE, \randX, \mat v, s)} \to \EV[F_0]{\varphi \left( \frac{\randE}{s} \right)}$ as $n \to \infty$ by the dominated convergence theorem.
Because $\rho$ has bounded second derivative, the derivative of $f_n(\randE, \randX, \mat v, s)$ with respect to $\mat v \in K$ and $s \in [\omscaleb - \epsilon_1, \omscaleb + \epsilon_1]$ is also bounded, meaning $f_n(\randE, \randX, \mat v, s)$ is equicontinuous on this domain.
Pointwise convergence together with the equicontinuity make the Arzelà-Ascoli theorem applicable and hence conclude that~\eqref{eqn:lemma3-uniform-conv} holds.

From \eqref{eqn:lemma3-statement-scale} it follows that for any $\delta_2 > 0$ there is a $N_{\delta_2}$ such that for all $\mat v \in K$ and all $n > N_{\delta_2}$, $\mathbb{P}\left( | \mhscaleast - \omscaleb | \leq \epsilon_1 \right) > 1 - \delta_2$. Combined with \eqref{eqn:lemma3-uniform-conv} this yields that for every $\delta_2 > 0$ and $\epsilon_2 > 0$ there is an $N_{\delta_2,\epsilon_2}$ such that for all $n > N_{\delta_2,\epsilon_2}$ and every $\mat v \in K$
\begin{displaymath}
\left|
	\EV[H_0]{f_n(\randE, \randX, \mat v, \mhscaleast)} -
	\EV[F_0]{\varphi \left( \frac{\randE}{ \mhscaleast } \right) }
\right| < \epsilon_2
\end{displaymath}
with probability greater than $1 - \delta_2$.
Since both expected values are positive this can also be written as
\begin{equation}\label{eqn:lemma3-denom-lb-1}
	\EV[H_0]{f_n(\randE, \randX, \mat v, \mhscaleast) } > \EV[F_0]{\varphi \left( \frac{u}{ \mhscaleast } \right) } - \epsilon_2.
\end{equation}

The final piece for the denominator to be bounded is to show that
\begin{equation}\label{eqn:lemma3-denom-lb-2}
\sup_{\mat v \in K} \left|
	\EV[H_0]{\varphi \left( \frac{\randE}{ \mhscaleast } \right) } -
	\EV[F_0]{\varphi \left( \frac{\randE}{\omscaleb} \right) }
\right| \almostsuren 0.
\end{equation}
Set $\Omega_1 = \{ \omega: \mhscale{\mat \beta_n^*; \omega} \to \omscaleb \}$ which has $\mathbb{P}(\Omega_1) = 1$ due to the first part of this lemma.
Similarly, set $\Omega_2 =  \{ \omega: \text{ equation \eqref{eqn:lemma3-denom-lb-2} holds} \}$.
Assume now that $\mathbb{P}(\Omega_1 \cap \Omega_2^\mathsf{c}) > 0$.
This assumption entails that there exists an $\omega' \in \Omega_1 \cap \Omega_2^\mathsf{c}$, an $\epsilon_3 > 0$ and a subsequence $(n_k)_{k>0}$ such that
\begin{equation}\label{eqn:lemma3-denom-false}
\lim_{k \to \infty} \left|
	\EV[H_0]{\varphi \left( \frac{\randE}{ \mhscale{ \btrue + \frac{ \mat v_{n_k} } { \sqrt{n_k} }; \omega' } } \right) } -
	\EV[F_0]{\varphi \left( \frac{\randE}{\omscaleb} \right) }
\right| > \epsilon_3.
\end{equation}
However, since $\mat v_{n_k}$ is in the compact set $K$, the sequence $\btrue + \mat v_{n_k} / \sqrt{n_k}$ converges to $\btrue$ as $n \to \infty$.
Additionally, $\varphi$ is bounded and together with the dominated convergence theorem this leads to
\begin{displaymath}
\lim_{k \to \infty} \EV[H_0]{\varphi \left( \frac{\randE}{ \mhscale{ \btrue + \mat v_{n_k} / \sqrt{n_k}; \omega' } } \right) }
	= \EV[F_0]{\varphi \left( \frac{\randE}{ \omscaleb } \right) }
\end{displaymath}
and in turn to
\begin{displaymath}
\lim_{k \to \infty} \left|
	\EV[H_0]{\varphi \left( \frac{\randE}{ \mhscale{ \btrue + \mat v_{n_k} / \sqrt{n_k}; \omega' } } \right)  } -
	\EV[F_0]{\varphi \left( \frac{\randE}{\omscaleb} \right) }
\right| = 0
\end{displaymath}
contradicting the claim in \eqref{eqn:lemma3-denom-false}.
Therefore, $\mathbb{P}(\Omega_1 \cap \Omega_2^\mathsf{c}) = 0$, proving \eqref{eqn:lemma3-denom-lb-2}.
Combining \eqref{eqn:lemma3-denom-lb-1} and \eqref{eqn:lemma3-denom-lb-2} leads to the conclusion that with arbitrarily high probability for large enough $n$
\begin{equation}\label{eqn:lemma3-denom}
\left|
	\EV[H_0]{\varphi \left( \frac{\randE -  \randX\tr \mat v / \sqrt{n}}{\mhscaleast} \right) }
\right| > -\epsilon_4 + \EV[F_0]{\varphi \left( \frac{\randE}{\omscaleb} \right) }
\end{equation}
for every $\mat v \in K$.

From the first part of this lemma, $\mhscaleast \almostsure \omscaleb$, and due to~\eqref{eqn:lemma3-denom}, for every $\delta > 0$ and every $0 < \epsilon < \EV[F_0]{\varphi \left( \frac{\randE}{\omscaleb} \right)}$ there exists an $N_{\delta, \epsilon}$ such that for all $\mat v \in K$ and $n \geq N_{\delta, \epsilon}$ equation~\eqref{eqn:lemma3-statement} holds.
\end{proof}

\subsection{Root-n consistency}\label{sec:appendix-proof-root-n-adapense}

\begin{proposition}\label{prop:strong-consistency}
Let $(y_i, \mat x_i\tr)$, $i = 1, \dotsc, n$, be i.i.d observations with distribution $H_0$ satisfying \eqref{ass:joint-distribution}.
Under assumptions \ref{ass:reg-1} and \ref{ass:reg-2}, PENSE and adaptive PENSE are both strongly consistent estimators of the true regression parameter $\trueparam$.
Specifically,
\begin{enumerate}[label=(\roman*)] 
\item\label{prop:strong-consistency-pense}
if $\lambdap_n \to 0$, the PENSE estimator $\pense$ as defined in~\eqref{def:pense} satisfies $\pense \xrightarrow{a.s.} \trueparam$;

\item\label{prop:strong-consistency-adapense}
if $\lambdaadap_n \to 0$, the adaptive PENSE estimator $\adapense$ as defined in~\eqref{def:adaptive-pense} satisfies $\adapense \xrightarrow{a.s.} \trueparam$.
\end{enumerate}
\end{proposition}

For part~\ref{prop:strong-consistency-pense}, the proof is identical to the proof of  strong consistency for the S-Ridge estimator (Proposition~1.i) in \textcite{Smucler2017} and hence omitted.
Although the penalty functions used for the S-Ridge and PENSE are different, the growth condition on $\lambdap_n$ has the same effect on PENSE as on the S-Ridge; making the penalty term negligible for large enough $n$.

Similarly, for part~\ref{prop:strong-consistency-adapense}, noting that the level of $L_2$ penalization given by $\lambdaadap_n (1 - \alphaadap) / 2$ converges deterministically to 0, the proof of strong consistency of adaptive PENSE is otherwise identical to the proof of strong consistency of adaptive MM-LASSO given in Smucler and Yohai (2017).

\begin{proposition}\label{prop:root-n-consistency}
Let $(y_i, \mat x_i\tr)$, $i = 1, \dotsc, n$, be i.i.d.~observations with distribution $H_0$ satisfying \eqref{ass:joint-distribution}.
Under assumptions \ref{ass:reg-1}--\ref{ass:reg-3} PENSE and adaptive PENSE are both root-n consistent estimators of the true regression parameter $\trueparam$. Specifically,
\begin{enumerate}[label=(\roman*)] 
\item\label{prop:root-n-consistency-pense}
if $\lambdap_n \to 0$ and $\lambdap_n = O(1 / \sqrt{n})$, the PENSE estimator $\pense$ as defined in~\eqref{def:pense}, satisfies: $\| \pense - \trueparam \| = O_p(1/\sqrt{n})$;

\item\label{prop:root-n-consistency-adapense}
if $\lambdaadap_n \to 0$ and $\lambdaadap_n = O(1 / \sqrt{n})$, the adaptive PENSE estimator $\adapense$ as defined in~\eqref{def:adaptive-pense}, satisfies: $\| \adapense - \trueparam \| = O_p(1/\sqrt{n})$.
\end{enumerate}
\end{proposition}

\begin{proof}
We will prove part~\ref{prop:root-n-consistency-adapense}, as part~\ref{prop:root-n-consistency-pense} is essentially the same.
The only difference is that the penalty loadings are deterministic and fixed at $(1, \dotsc, 1)\tr$ leading to a different constant $D$ below.

To ease the notation for the proof, the adaptive elastic net penalty is simply denoted by $\Phi(\mat\beta) = \Phi_\subt{AE}(\mat \beta; \lambdaadap, \alphaadap, \zeta, \mat\omega)$.
Also, $\mat \gamma(t) := \btrue + t (\badapense - \btrue)$ denotes the convex combination of the true parameter $\btrue$ and the adaptive PENSE estimator $\badapense$.

The first step in the proof is a Taylor expansion of the objective function around the true parameter $\btrue$:
\begin{align*}
\mhscaleadasq + \Phi(\badapense) = & \omhscalesqb + \Phi(\btrue) + (\Phi(\badapense) - \Phi(\btrue)) \\
	&- 2 \underbrace{ \frac{1} { \frac{1}{n} \sum_{i = 1}^{n} \varphi\left( \frac{u_i - \mat x_i\tr \mat v_n} {\mhscaleast} \right) } }_{=: A_n}
		\underbrace{ \frac{ \mhscaleast }{n} \sum_{i = 1}^n \psi\left( \frac{u_i - \mat x_i\tr \mat v_n} {\mhscaleast} \right) \mat x_i\tr \mat v_n }_{=: Z_n}
\end{align*}
where $\mat  v_n = \tau (\badapense - \btrue)$ and $\mat\beta_n^* = \btrue + \mat v_n$ for a $0 < \tau < 1$.
Due to the strong consistency of $\badapense$ from Proposition~\ref{prop:strong-consistency}, $\mat v_n \to 0$ a.s.\ and hence from Lemma~\ref{lem:as-convergence-vn-as} and the continuous mapping theorem it is know that $A_n \almostsure \frac{ 1 } { \EV[F_0]{\varphi \left( \frac{\randE}{\omscaleb} \right) } } =: A > 0$ as well as $\mhscaleast \almostsure \omscaleb$.
The term $Z_n$ is handled by a Taylor expansion of $\psi\left( \frac{u_i - \mat x_i\tr \mat v_n } {\mhscaleast} \right)$ around $u_i$ to get
\begin{align*}
Z_n &= \mhscaleast \left(
	\frac{1}{n} \sum_{i = 1}^n \psi\left( \frac{u_i } {\mhscaleast} \right) \mat x_i\tr \mat v_n -
		\frac{1}{ \mhscaleast n} \sum_{i = 1}^n \psi\left( \frac{u_i - \mat x_i\tr \mat v^*_n} {\mhscaleast} \right) \mat x_i\tr \mat v_n \mat x_i\tr \mat v_n
	\right) \\
	&= \frac{ (\badapense - \btrue)\tr } {\sqrt{n}} \left[ \tau \mhscaleast \frac{1}{\sqrt{n}} \sum_{i = 1}^n \psi\left( \frac{u_i } {\mhscaleast} \right) \mat x_i \right] \\
	&\quad - \tau^2 (\badapense - \btrue)\tr \left[ \frac{1}{n} \sum_{i = 1}^n \psi' \left( \frac{u_i - \mat x_i\tr \mat v^*_n} {\mhscaleast} \right) \mat x_i \mat x_i\tr \right] (\badapense - \btrue)
\end{align*}
for some $\mat v^*_n = \tau^* \mat v_n$ with $\tau^* \in (0, 1)$.

The rest of the proof follows closely the proof of Proposition~2 in \textcite{Smucler2017}.
More specifically, noting that $\mhscaleast \almostsure \omscaleb$, the results in \textcite{Smucler2017} (which are derived from results in \textcite{Yohai1985}) state that
\begin{displaymath}
B_n := \left\| \mat\xi_n \right\| = O_p(1)
\quad\text{with}\quad
\mat\xi_n = \tau \mhscaleast \frac{1}{\sqrt{n}} \sum_{i = 1}^n \psi\left( \frac{u_i } {\mhscaleast} \right) \mat x_i
\end{displaymath}
and hence with arbitrarily high probability for $n$ sufficiently large there is a $B$ such that
\begin{equation}\label{eqn:root-n-consistency-bound-bn}
\frac{ (\badapense - \btrue)\tr } {\sqrt{n}} \mat\xi_n \leq \frac{ 1 } {\sqrt{n}} \| \badapense - \btrue \| \| \mat\xi_n\| \leq \frac{ B } {\sqrt{n}} \| \badapense - \btrue \|.
\end{equation}
Similarly, the results in \textcite{Smucler2017} can be used to show
\begin{equation}\label{eqn:root-n-consistency-bound-cn}
C_n := \tau^2 (\badapense - \btrue)\tr \left[ \frac{1}{n} \sum_{i = 1}^n \psi' \left( \frac{u_i - \mat x_i\tr \mat v^*_n} {\mhscaleast} \right) \mat x_i \mat x_i\tr \right] (\badapense - \btrue) \geq \widetilde C_n \| \badapense - \btrue \|^2
\end{equation}
with $\widetilde C_n \almostsure C > 0$.

Next is the difference in the penalty terms $D_n := \Phi(\badapense) - \Phi(\btrue)$, which can be reduced to the truly non-zero coefficients:
\begin{align*}
D_n =&  \lambdaadap_{n} \sum_{j = 1}^p | \bpenseel_j |^{-\zeta} \left( \frac{1 - \alpha}{2}  \left( (\badapenseel_j)^2 - (\beta^0_j)^2 \right) + \alpha ( |\badapenseel_j| - |\beta^0_j| ) \right) \\
	\geq& \lambdaadap_{n} \sum_{j = 1}^s | \bpenseel_j |^{-\zeta} \left( \frac{1 - \alpha}{2}  \left( (\badapenseel_j)^2 - (\beta^0_j)^2 \right) + \alpha ( |\badapenseel_j| - |\beta^0_j| ) \right).
\end{align*}
Observing that $\badapense$ is a strongly consistent estimator, $| \badapenseel_j - \beta^0_j | < \epsilon_j < | \beta^0_j |$ for all $j = 1, \dotsc, s$ and any $\epsilon_j \in (0, | \beta^0_j |)$ with arbitrarily high probability for sufficiently large $n$.
This entails that, for all $0 \leq t \leq 1$ and $j = 1, \dotsc, s$, the sign of the convex combination $\operatorname{sgn}(\gamma_j(t)) = \operatorname{sgn}(\beta^0_j) \neq 0$ and thus $| \gamma_j(t) |$ is differentiable.
This allows application of the mean value theorem on the quadratic and the absolute term in $D_n$ to yield
\begin{align*}
D_n \geq&  \lambdaadap_{n} \sum_{j = 1}^s | \bpenseel_j |^{-\zeta} \left( \frac{1 - \alpha}{4} \gamma_j(\tau_j) + \alpha  \operatorname{sgn}(\beta^0_j) \right) (\badapenseel_j - \beta^0_j)
\end{align*}
for some $\tau_j \in (0, 1)$, $j = 1, \dotsc, s$, with arbitrarily high probability for large enough $n$.
Because both $\bpense$ and $\badapense$ are strongly consistent for $\btrue$ and $\lambdaadap_{n} = O(1 / \sqrt{n})$, there exists a constant $D$ such that with arbitrarily high probability
\begin{equation}\label{eqn:root-n-consistency-bound-dn}
D_n \geq - \frac{D}{\sqrt{n}} \| \badapense - \btrue \|
\end{equation}
for sufficiently large $n$.

Since $\badapense$ minimizes the adaptive PENSE objective function $\objf[AS]$,
\begin{align*}
0 \geq& \objf[AS](\badapense) - \objf[AS](\btrue)
	= \mhscaleadasq + \Phi(\badapense) - \omhscalesqb - \Phi(\btrue)
	= D_n - 2 A_n Z_n.
\end{align*}
With the bounds derived in \eqref{eqn:root-n-consistency-bound-bn}, \eqref{eqn:root-n-consistency-bound-cn}, and \eqref{eqn:root-n-consistency-bound-dn} this in turn yields
\begin{align*}
0 \geq& D_n - 2 A_n Z_n = D_n - 2 A_n B_n + 2 A_n C_n \\
	\geq& - \frac{D}{\sqrt{n}} \| \badapense - \btrue \| - 2 A \frac{ B } {\sqrt{n}}  \| \badapense - \btrue \| + 2 A C \| \badapense - \btrue \|^2 \\
	=& \frac{1}{\sqrt{n}} \| \badapense - \btrue \| \left( - D - 2 A B + 2 A C \sqrt{n} \| \badapense - \btrue \| \right)
\end{align*}
with arbitrarily high probability for large enough $n$.
Rearranging the terms leads to the inequality
\begin{displaymath}
\sqrt{n} \| \badapense - \btrue \| \leq \frac{2 A B + D}{2 A C}.
\end{displaymath}
\end{proof}

\subsection{Variable selection consistency}\label{sec:appendix-proof-varsel-cons-adapense}

\begin{proof}[Proof of Theorem~\ref{thm:asymptotic-properties}, part~\ref{thm:var-sel-consistency}]
To ease notation in the following, we denote the coordinate-wise adaptive EN penalty function by $$
\phi(\beta; \lambdaadap_{n}, \alphaadap, \zeta, \tilde\beta) = \lambdaadap_{n} |\tilde\beta |^{-\zeta} \left( \frac{1 - \alphaadap}{2} \beta^2 + \alphaadap | \beta | \right)
$$
such that $\lambdaadap_{n} \Phi_\subt{AE}(\mat\beta; \alphaadap, \zeta, \bpense) = \sum_{j=1}^p \phi(\beta_j; \lambdaadap_{n}, \alphaadap, \zeta, \bpenseel_j)$.
We follow the proof in \textcite{Smucler2017} and define the function
\begin{align*}
V_n(\mat v_1, \mat v_2) := & \mhscalesq{\btrue[\subRnum{1}] + \mat v_1 / \sqrt{n}, \btrue[\subRnum{2}]+ \mat v_2 / \sqrt{n}} + \\
	& \sum_{j=1}^s \phi(\beta^0_j + v_{1,j} / \sqrt{n}; \lambdaadap_{n}, \alphaadap, \zeta, \bpenseel_j) + \\
	& \sum_{j=s + 1}^p \phi(\beta^0_j + v_{2,j-s} / \sqrt{n}; \lambdaadap_{n}, \alphaadap, \zeta, \bpenseel_j).
\end{align*}
From Proposition~\ref{prop:root-n-consistency} follows with arbitrarily high probability, $\| \badapense - \btrue \| \leq C / \sqrt{n}$ for sufficiently large $n$.
Therefore, with arbitrarily high probability $V_n(\mat v_1, \mat v_2)$ attains its minimum on the compact set $\left\{ (\mat v_1, \mat v_2) : \|\mat v_1\|^2 + \|\mat v_2\|^2 \leq C^2 \right\}$ at $\badapense$.
The goal is to show that for any $ \|\mat v_1\|^2 + \|\mat v_2\|^2 \leq C^2$ with $\| \mat v_2 \| > 0$ and with arbitrarily high probability, $V_n(\mat v_1, \mat v_2) - V_n(\mat v_1, \mat 0_{p-s}) > 0$ for sufficiently large $n$.

Taking the difference while observing that $\btrue[\subRnum{2}] = \mat 0_{p-s}$ gives
\begin{align*}
V_n(\mat v_1, \mat v_2) - V_n(\mat v_1, \mat 0_{p-s}) = & 
	\left( \mhscalesq{\btrue[\subRnum{1}] + \mat v_1 / \sqrt{n}, \mat v_2 / \sqrt{n}} - \mhscalesq{\btrue[\subRnum{1}] + \mat v_1 / \sqrt{n}, \mat 0_{p -s}} \right) + \\
	& \sum_{j=s + 1}^p \phi(v_{2,j-s} / \sqrt{n}; \lambdaadap_{n}, \alphaadap, \zeta, \bpenseel_j).
\end{align*}
The first term can be bounded by defining $\mat v_n(t) :=  (\mat v_1\tr, t \mat v_2\tr)\tr / \sqrt{n}$ and applying the mean value theorem gives some $\tau \in (0, 1)$ such that
\begin{align*}
\mhscalesq{\btrue + \mat v_n(1)} - \mhscalesq{\btrue + \mat v_n(0)} = & \\
	\frac{2}{\sqrt{n}} \mhscale{\btrue + \mat v_n(\tau)} (\mat 0_s\tr, \mat v_2\tr) \subgrad{\mat\beta}\mhscale{\mat\beta} \evalat{\btrue + \mat v_n(\tau)} = & \\
	- \frac{2}{\sqrt{n}}
		\underbrace{ \frac{\mhscale{\btrue + \mat v_n(\tau)}} {\frac{1}{n} \sum_{i = 1}^n \varphi \left( \frac{u_i - \mat x_i\tr \mat v_n(\tau)} {\mhscale{\btrue + \mat v_n(\tau)}} \right) } }_{=: A_n}
		\underbrace{ (\mat 0_s\tr, \mat v_2\tr) \frac{1}{n} \sum_{i = 1}^n \psi \left( \frac{u_i - \mat x_i\tr \mat v_n(\tau)} {\mhscale{\btrue + \mat v_n(\tau)}} \right) \mat x_i}_{=: B_n} &.
\end{align*}
By Lemma~\ref{lem:uniform-convergence_v_compact} the term $A_n$ is uniformly bounded in probability, hence $|A_n| < A$ with arbitrarily high probability for large enough $n$.
Furthermore, $|B_n| \leq \| \psi \|_\infty \| \mat v_2 \|  \left\| \frac{1}{n} \sum_{i = 1}^n \mat x_i \right\|$ and due to the law of large numbers there is a constant $B$ such that the upper bound for $|B_n|$ is
\begin{displaymath}
|B_n| \leq \| \psi \|_\infty \| \mat v_2 \|  ( \| \EV[G_0]{\randX} \| + \epsilon) < \| \mat v_2 \| B
\end{displaymath}
with arbitrarily high probability for sufficiently large $n$.
Together, the bounds for $A_n$ and $B_n$ give
\begin{equation}\label{eqn:thm-var-sel-consistency-lossdiff-bound}
\mhscalesq{\btrue + \mat v_n(1)} - \mhscalesq{\btrue + \mat v_n(0)} \geq 
	- \frac{\| \mat v_2 \|}{\sqrt{n}} 2 A B.
\end{equation}

The next step is to ensure that the penalty term grows large enough to make the difference $V_n(\mat v_1, \mat v_2) - V_n(\mat v_1, \mat 0_{p-s})$ positive.
Indeed, the assumption $\alphaadap > 0$ leads to
\begin{align*}
\sum_{j=s + 1}^p \phi(v_{2,j-s} / \sqrt{n}; \lambdaadap_{n}, \alphaadap, \zeta, \omega_j)
	&\geq \alphaadap \lambdaadap_{n} \sum_{j=s + 1}^p | \bpenseel_j |^{-\zeta} \frac{ | v_{2, j - s} | } { \sqrt{n} } \\
	&= \alphaadap \lambdaadap_{n} n^{(\zeta -1)/2} \sum_{j=s + 1}^p | \sqrt{n} \bpenseel_j |^{-\zeta} | v_{2, j - s} |.
\end{align*}
The root-n consistency of $\bpense$ established in Proposition~\ref{prop:root-n-consistency} gives $|\sqrt{n} \bpenseel_j | < M$ with arbitrarily high probability for large enough $n$.
Therefore,
\begin{equation}\label{eqn:thm-var-sel-consistency-penalty-bound}
\begin{aligned}
\alphaadap \lambdaadap_{n} n^{(\zeta -1)/2} \sum_{j=s + 1}^p \frac{ | v_{2, j - s} | } { | \sqrt{n} \bpenseel_j |^\zeta }
	&> \alphaadap \lambdaadap_{n} n^{(\zeta -1)/2} \sum_{j=s + 1}^p M^{-\zeta} | v_{2, j - s} | \\
	&= \alphaadap\lambdaadap_{n} n^{(\zeta -1)/2} M^{-\zeta} \| v_2 \|_1 \\
	&\geq \frac{ \| v_2 \| } { \sqrt{n} } M^{-\zeta} \alphaadap \lambdaadap_{n} n^{\zeta/2}.
\end{aligned}
\end{equation}
Combining~\eqref{eqn:thm-var-sel-consistency-lossdiff-bound} and~\eqref{eqn:thm-var-sel-consistency-penalty-bound} yields
\begin{equation}\label{eqn:thm-var-sel-consistency-combined-bound}
V_n(\mat v_1, \mat v_2) - V_n(\mat v_1, \mat 0_{p-s}) >
	\frac{\| \mat v_2 \|}{\sqrt{n}} \left( - 2 A B + M^{-\zeta} \alphaadap \lambdaadap_{n} n^{\zeta/2} \right)
\end{equation}
uniformly over $\mat v_1$ and $\mat v_2$ with arbitrarily high probability for sufficiently large $n$.
By assumption $\alphaadap \lambdaadap_{n} n^{\zeta/2} \to \infty$ and hence the right-hand side in \eqref{eqn:thm-var-sel-consistency-combined-bound} will eventually be positive, concluding the proof.
\end{proof}

\subsection{Asymptotic Normal distribution}\label{sec:appendix-proof-adapense-normal}
\begin{proof}[Proof of Theorem~\ref{thm:asymptotic-properties}, part~\ref{thm:asymptotic-normality}]
For this proof we denote the values of the active predictors and the active predictors in the $i$-th observation by $\mat X_\subRnum{1}$ and $\mat x_{i,\subRnum{1}}$, respectively.
Because $\badapense$ is strongly consistent for $\btrue$, the coefficient values for the truly active predictors are almost surely bounded away from zero if $n$ is large enough.
This entails that the partial derivatives of the penalty function exist for the truly active predictors and the gradient at the estimate $\badapense$ is
\begin{equation}
\label{eqn:asymnorm-partial-deriva}
\mat 0_s = \subgrad{\mat\beta_\subRnum{1}} \objf[AS] (\badapense) =
- 2 \frac{\mhscaleada} {A_n} \frac{1}{n} \sum_{i=1}^n \psi \left( \frac{ y_i - \mat x_i\tr \badapense} {\mhscaleada}  \right) \mat x_{i,\subRnum{1}}
+ \subgrad{\mat\beta_\subRnum{1}} \Phi_\subt{AE} (\badapense; \lambdaadap_{n}, \alphaadap, \zeta, \bpense)
\end{equation}
with $A_n = \frac{1}{n} \sum_{i=1}^n \varphi \left( \frac{ y_i - \mat x_i\tr \badapense} {\mhscaleada}  \right)$.
The truly active coefficients can be separated from the truly inactive coefficients by noting that $\psi \left( \frac{ y_i - \mat x_i\tr \badapense} {\mhscaleada}  \right) = \psi \left( \frac{ y_i - \mat x_{i,\subRnum{1}}\tr \badapense[\subRnum{1}]} {\mhscaleada} \right) + o_i$ for some $o_i$ which vanishes in probability, $\mathbb{P}(o_i = 0) \to 1$, because of Theorem~\ref{thm:asymptotic-properties}, part~\ref{thm:var-sel-consistency} and because $\psi$ is continuous.
Equation~\eqref{eqn:asymnorm-partial-deriva} can now be written as
\begin{displaymath}
\begin{aligned}
\mat 0_s =&
- 2 \frac{\mhscaleada} {A_n} \frac{1}{\sqrt{n}} \sum_{i=1}^n \psi \left( \frac{ y_i - \mat x_{i,\subRnum{1}}\tr \badapense_{\subRnum{1}}} {\mhscaleada}  \right) \mat x_{i,\subRnum{1}} \\
&\quad- 2 \frac{\mhscaleada} {A_n} \frac{1}{\sqrt{n}} \sum_{i=1}^n o_i \mat x_{i,\subRnum{1}} \\
&\quad + \sqrt{n} \subgrad{\mat\beta_\subRnum{1}} \Phi_\subt{AE} (\badapense; \lambdaadap_{n}, \alphaadap, \zeta, \bpense)
\end{aligned}
\end{displaymath}
and using the mean value theorem there are $\tau_i \in [0, 1]$ and hence a matrix
$$
\mat W_n = \frac{1}{n} \sum_{i=1}^n \psi' \left( \frac{ u_i - \tau_i \mat x_{i,\subRnum{1}}\tr \left( \badapense[\subRnum{1}] - \btrue[\subRnum{1}] \right) } {\mhscaleada}  \right) \mat x_{i,\subRnum{1}} \mat x_{i,\subRnum{1}}\tr$$
such that the equation can be further rewritten to
\begin{displaymath}
\begin{aligned}
\mat 0_s &=
- 2 \frac{\mhscaleada} {A_n} \frac{1}{\sqrt{n}} \sum_{i=1}^n \psi \left( \frac{ y_i - \mat x_{i,\subRnum{1}}\tr \btrue[\subRnum{1}]} {\mhscaleada}  \right) \mat x_{i,\subRnum{1}} \\
&\quad + 2 \frac{1} {A_n} \mat W_n \sqrt{n} \left( \badapense[\subRnum{1}] - \btrue[\subRnum{1}] \right) \\
&\quad - 2 \frac{\mhscaleada} {A_n} \frac{1}{\sqrt{n}} \sum_{i=1}^n o_i \mat x_{i,\subRnum{1}} \\
&\quad + \sqrt{n} \lambdaadap_{n} \subgrad{\mat\beta_\subRnum{1}} \Phi_\subt{AE} (\badapense; \alphaadap, \zeta, \bpense).
\end{aligned}
\end{displaymath}
Separating the term $\sqrt{n} \left( \badapense[\subRnum{1}]^* - \btrue[\subRnum{1}] \right)$ then gives
\begin{equation}
\label{eqn:asymnorm-final-form}
\begin{aligned}
\sqrt{n} \left( \badapense[\subRnum{1}]^* - \btrue[\subRnum{1}] \right) &=
\mhscaleada \mat W_n^{-1} \frac{1}{\sqrt{n}} \sum_{i=1}^n \psi \left( \frac{ y_i - \mat x_{i,\subRnum{1}}\tr \btrue[\subRnum{1}]} {\mhscaleada}  \right) \mat x_{i,\subRnum{1}} \\
&\quad + \mhscaleada \mat W_n^{-1} \frac{1}{\sqrt{n}} \sum_{i=1}^n o_i \mat x_{i,\subRnum{1}} \\
&\quad + \sqrt{n} \lambdaadap_{n} \mhscaleada A_n \mat W_n^{-1} \subgrad{\mat\beta_\subRnum{1}} \Phi_\subt{AE} (\badapense; \alphaadap, \zeta, \bpense).
\end{aligned}
\end{equation}

The strong consistency of $\badapense$ for $\btrue$ and Lemma~\ref{lem:as-convergence-vn-as} lead to $\mhscaleada \xrightarrow{a.s.} \omscaleb$ and $A_n \xrightarrow{a.s.} \EV[F_0]{\varphi \left( \frac{\randE}{\omscaleb} \right) } < \infty$.
Also, because of $\mhscaleada \xrightarrow{a.s.} \omscaleb$, Lemma~4.2 in \textcite{Yohai1985}, and the law of large numbers
$$
\mat W_n \xrightarrow{a.s.} b(\rho, F_0) \mat \Sigma_{\subRnum{1}}.
$$
Combined with the assumption that $\sqrt{n} \lambdaadap_{n} \to 0$ this leads to the last two lines in \eqref{eqn:asymnorm-final-form} converging to $\mat 0_s$ in probability.
Finally by Lemma~5.1 in \textcite{Yohai1985} and the CLT
\begin{displaymath}
\frac{1}{\sqrt{n}} \sum_{i=1}^n \psi \left( \frac{ y_i - \mat x_{i,\subRnum{1}}\tr \btrue[\subRnum{1}]} {\mhscaleada}  \right) \mat x_{i,\subRnum{1}} \xrightarrow{~d~}
N_s \left( \mat 0_s, a(\rho, F_0) \mat \Sigma_\subRnum{1} \right)
\end{displaymath}
which, after applying Slutsky's Theorem, completes the proof.
\end{proof}

%% file: supp-sim.tex

\section{Additional simulation results}

Here we present additional plots for the scenario detailed in Section~\ref{sec:simulation} as well as an alternative scenario with more severe contamination.

For the first scenario, Figure~\ref{fig:simstudy-scenario_1-prediction_accuracy} shows plots for prediction accuracy including the least-squares EN and adaptive least-squares EN estimators for comparison.
Moreover, in Figure~\ref{fig:simstudy-scenario_1-prediction_accuracy-relative} we show the prediction accuracy relative to the prediction accuracy achieved by adaptive PENSE, i.e.,
\begin{equation}\label{eqn:simstudy-rel-pred_err}
\text{RPP}_m = \begin{cases}
\hat\tau_m / \hat\tau_\subt{Ada.\ PENSE} - 1 & \text{if } \hat\tau_\subt{Ada.\ PENSE} \leq \hat\tau_m \\
- \hat\tau_\subt{Ada.\ PENSE} / \hat\tau_m + 1 & \text{if } \hat\tau_\subt{Ada.\ PENSE} > \hat\tau_m
\end{cases}.
\end{equation}
Here $m$ is the estimation method (e.g., adaptive MM), and $\hat\tau_m$ is the scale of the prediction error achieved by method $m$.
Hence, a value of $\text{RPP}_m = 0.1$ means that the scale of the prediction error of method $m$ is 10\% larger than that of adaptive PENSE, and a value of $e_m = -0.1$ means that the scale of the prediction error of adaptive PENSE is 10\% higher than that of method $m$.
For Normal errors, $\hat\tau$ is the root mean squared prediction error, while for all other error distributions the $\tau$-size of the prediction errors is used.

For a more fine-grained picture of variable selection performance, we also show the sensitivity and specificity in Figure~\ref{fig:simstudy-scenario_1-sens_spec}.

\begin{figure}[ht]
  {\centering \includegraphics[width=1\linewidth]{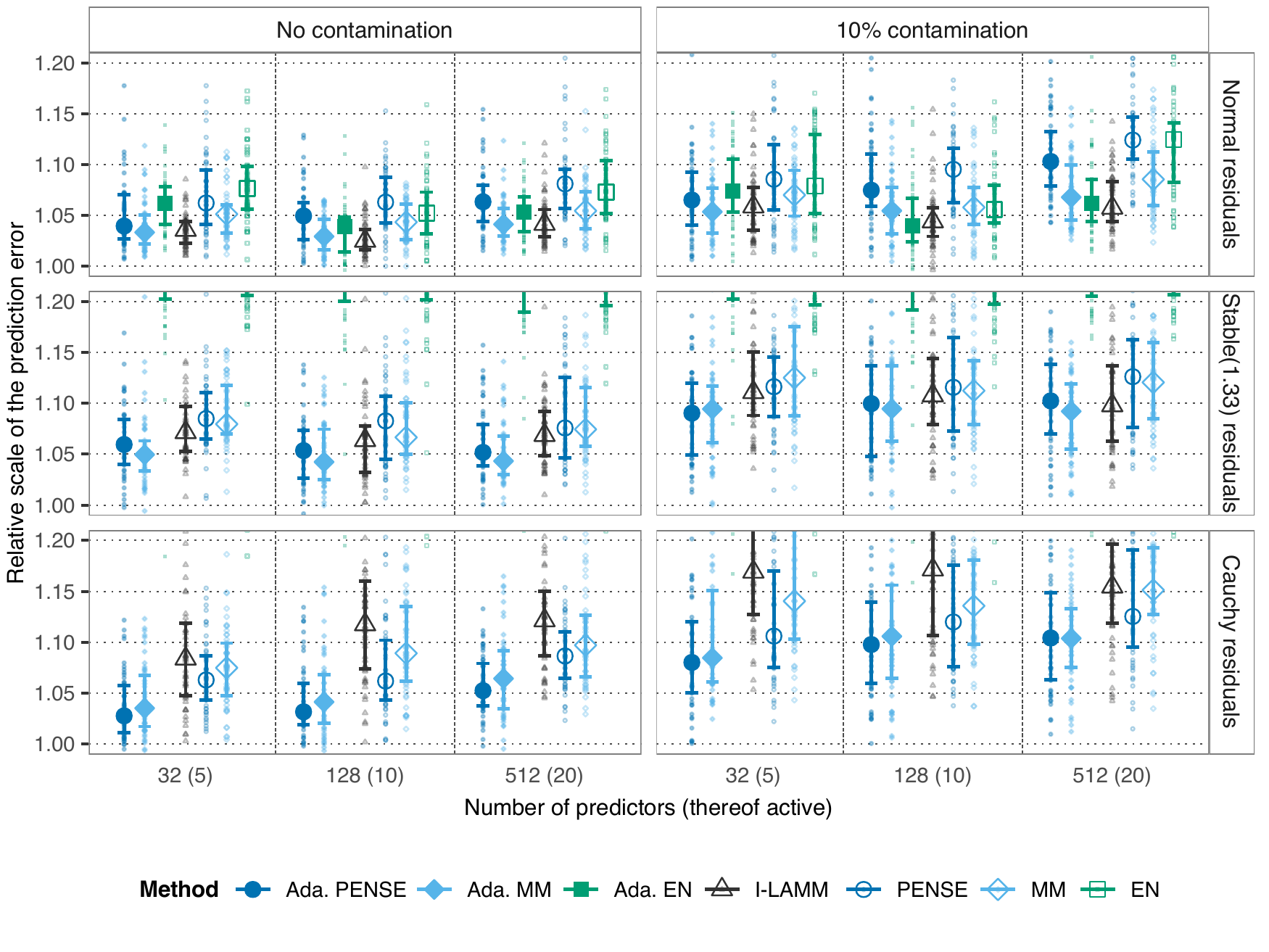}

  }
  \caption{%
  Prediction performance of robust and non-robust estimators, measured by the uncentered $\tau$-scale of the prediction errors relative to the true $\tau$-scale of the error distribution (lower is better).
  The median out of 50 replications is depicted by the points and the lines show the interquartile range.}
  \label{fig:simstudy-scenario_1-prediction_accuracy}
\end{figure}

\begin{figure}[ht]
  {\centering \includegraphics[width=1\linewidth]{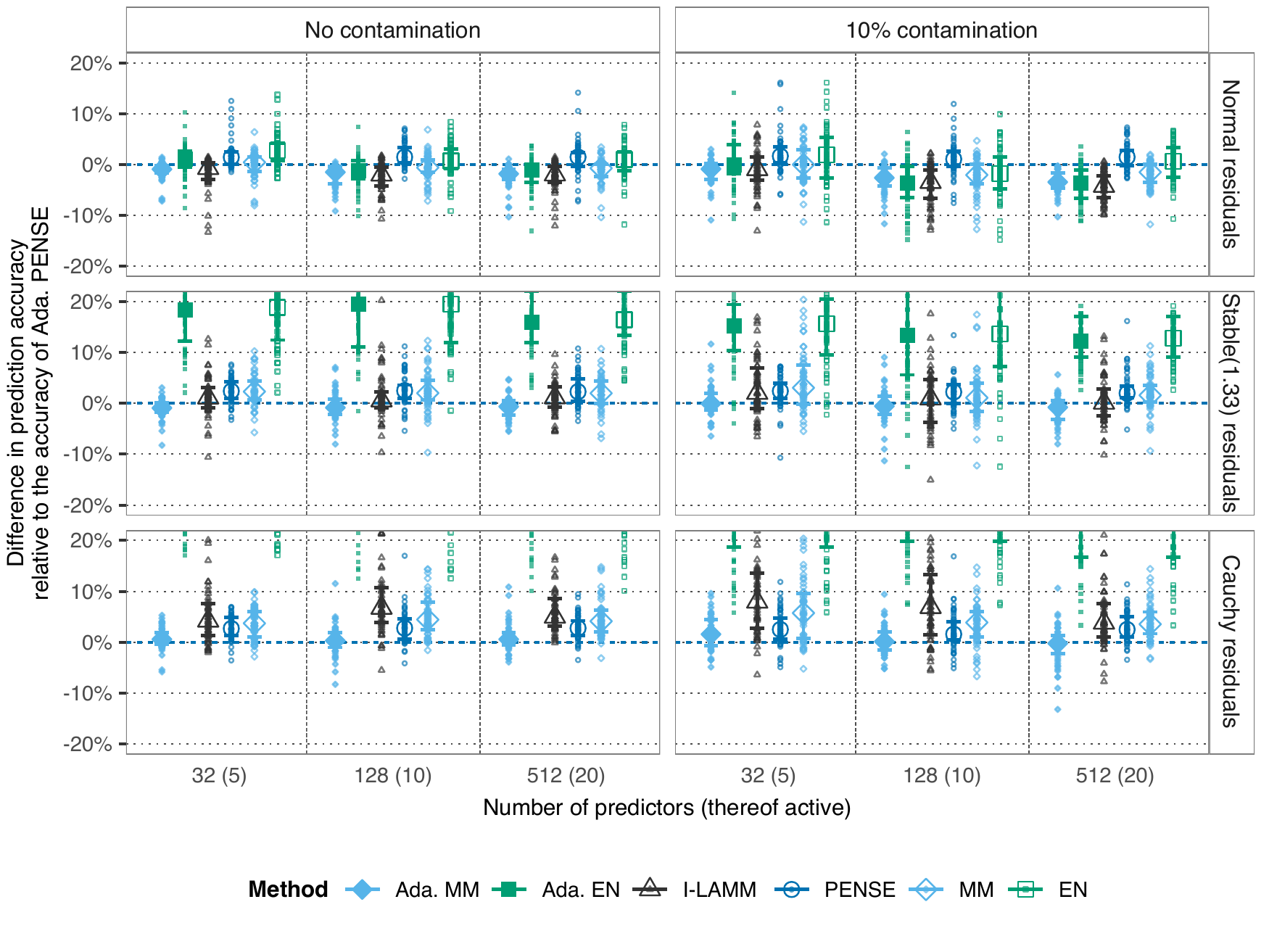}

  }
  \caption{%
  Prediction accuracy of robust and non-robust estimators in the alternative scenario, relative to the prediction accuracy of adaptive PENSE.
  The relative accuracy is defined in Equation~\eqref{eqn:simstudy-rel-pred_err}, with positive values indicating a larger scale of the prediction errors than achieved by adaptive PENSE and negative values indicating better prediction accuracy than adaptive PENSE.
The median out of 50 replications is depicted by the points and the lines show the interquartile range}
  \label{fig:simstudy-scenario_1-prediction_accuracy-relative}
\end{figure}

\begin{figure}[ht]
  {\centering \includegraphics[width=1\linewidth]{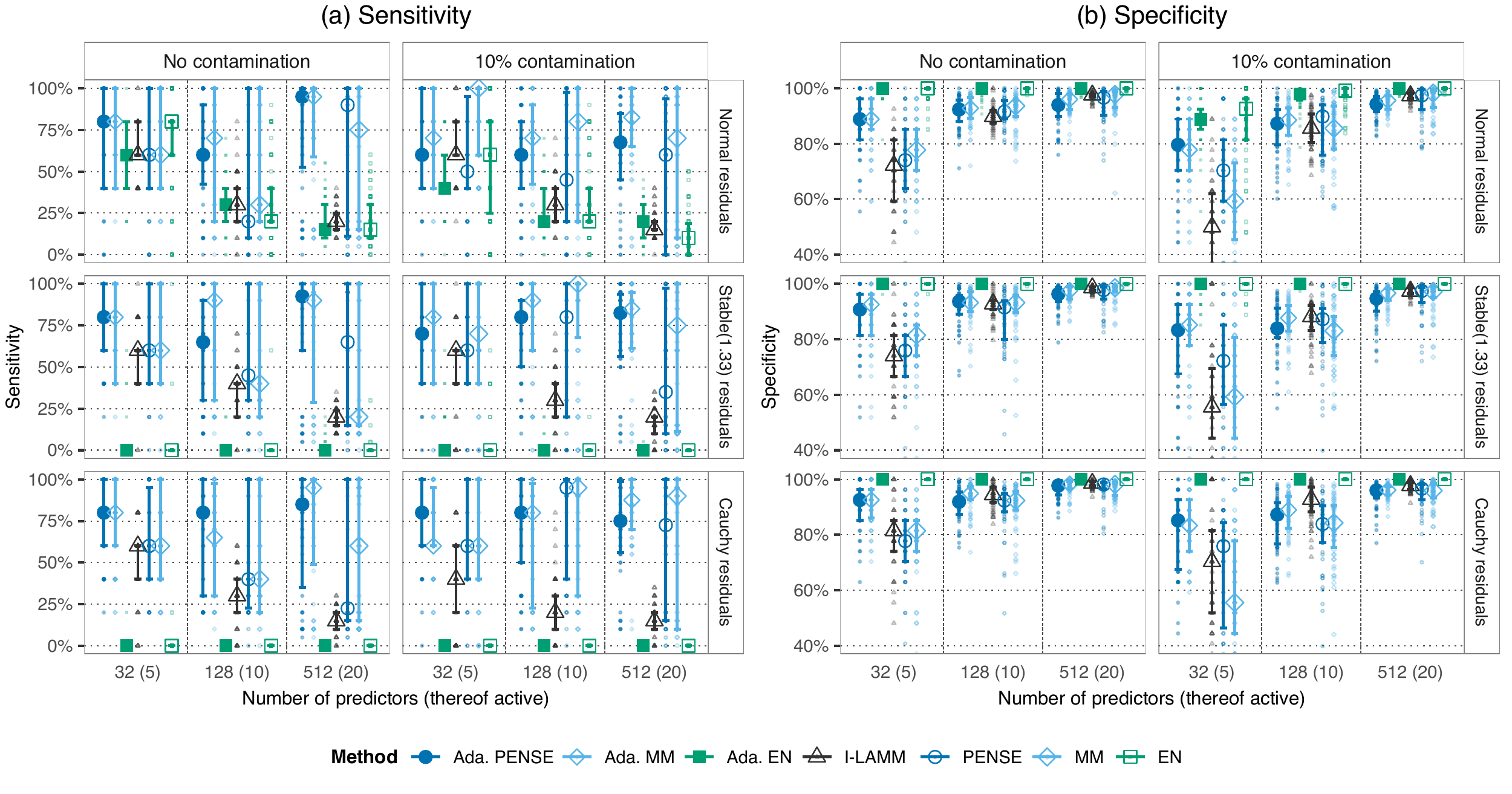}

  }
  \caption{%
  Variable selection performance of robust and non-robust estimators.
  Higher is better for both measures.
  Sensitivity (left) is the number of selected predictors which are truly active relative to the total number of truly active predictors.
  Specificity (right) is the number of predictors which are truly inactive and not selected relative to the total number of inactive predictors.
  The median out of 50 replications is depicted by the points and the lines show the interquartile range.}
  \label{fig:simstudy-scenario_1-sens_spec}
\end{figure}

\clearpage


\subsection{Alternative scenario}

Here we present additional results for a more challenging alternative scenario with $n=100$, $p=32$ to $p = 128$, and $s = \log_2(p)$.
The scenario introduces a higher proportion of contamination (25\%) and outliers are more severe ($k_\subt{lev} = 256$ and the contamination model uses coefficient value of $-3$ instead of $-1$).
In addition to the $\log_2(p)$ irrelevant predictors, bad leverage points also affect 2 truly relevant predictors.
Good leverage points have extreme values in at most $\log_2(p)$ non-relevant predictors.
The predictors follow a multivariate $t$-distribution with 4 degrees of freedom and AR1-type correlation structure of $\operatorname{Cor}(\randX_j, \randX_{j'}) = 0.5^{|j - j'|}$, $j, j' = 1, \dotsc, p$.

Figures~\ref{fig:simstudy-scenario_2-prediction_accuracy} and~\ref{fig:simstudy-scenario_2-prediction_accuracy-relative} show the prediction accuracy relative the the true error scale and relative to the prediction accuracy achieved by adaptive PENSE, respectively.
Overall variable selection performance in terms of the MCC is shown in Figure~\ref{fig:simstudy-scenario_2-mcc}, whereas Figure~\ref{fig:simstudy-scenario_2-sens_spec} gives more detailed insights into the sensitivity and specificity of variable selection performance.

Overall, the conclusions are very similar to the scenario presented in Section~\ref{sec:simulation}, showcasing that adaptive PENSE is highly robust and performance is similar to adaptive MM.
In this more challenging scenario, however we can observe that in some situations adaptive MM is affected considerable more by the contamination than adaptive PENSE, sometimes having more than 30\% higher prediction error than adaptive PENSE, as highlighted in Figure~\ref{fig:simstudy-scenario_2-prediction_accuracy-relative}.

\begin{figure}[ht]
{\centering \includegraphics[width=1\linewidth]{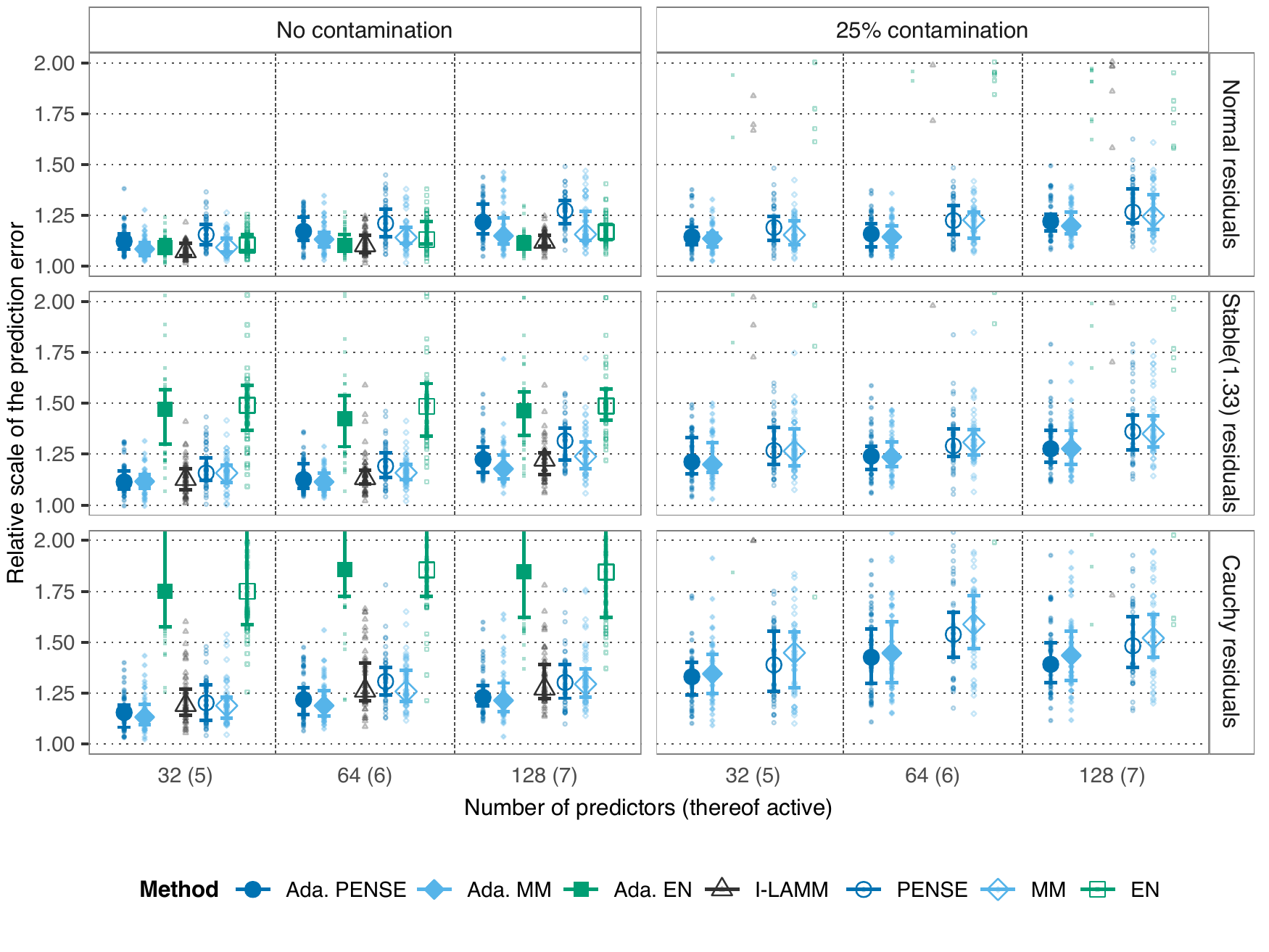}

}
\caption{%
Prediction performance of robust and non-robust estimators in the alternative scenario, measured by the uncentered $\tau$-scale of the prediction errors relative to the true $\tau$-scale of the error distribution (lower is better).
The median out of 50 replications is depicted by the points and the lines show the interquartile range.}
\label{fig:simstudy-scenario_2-prediction_accuracy}
\end{figure}

\begin{figure}[ht]
{\centering \includegraphics[width=1\linewidth]{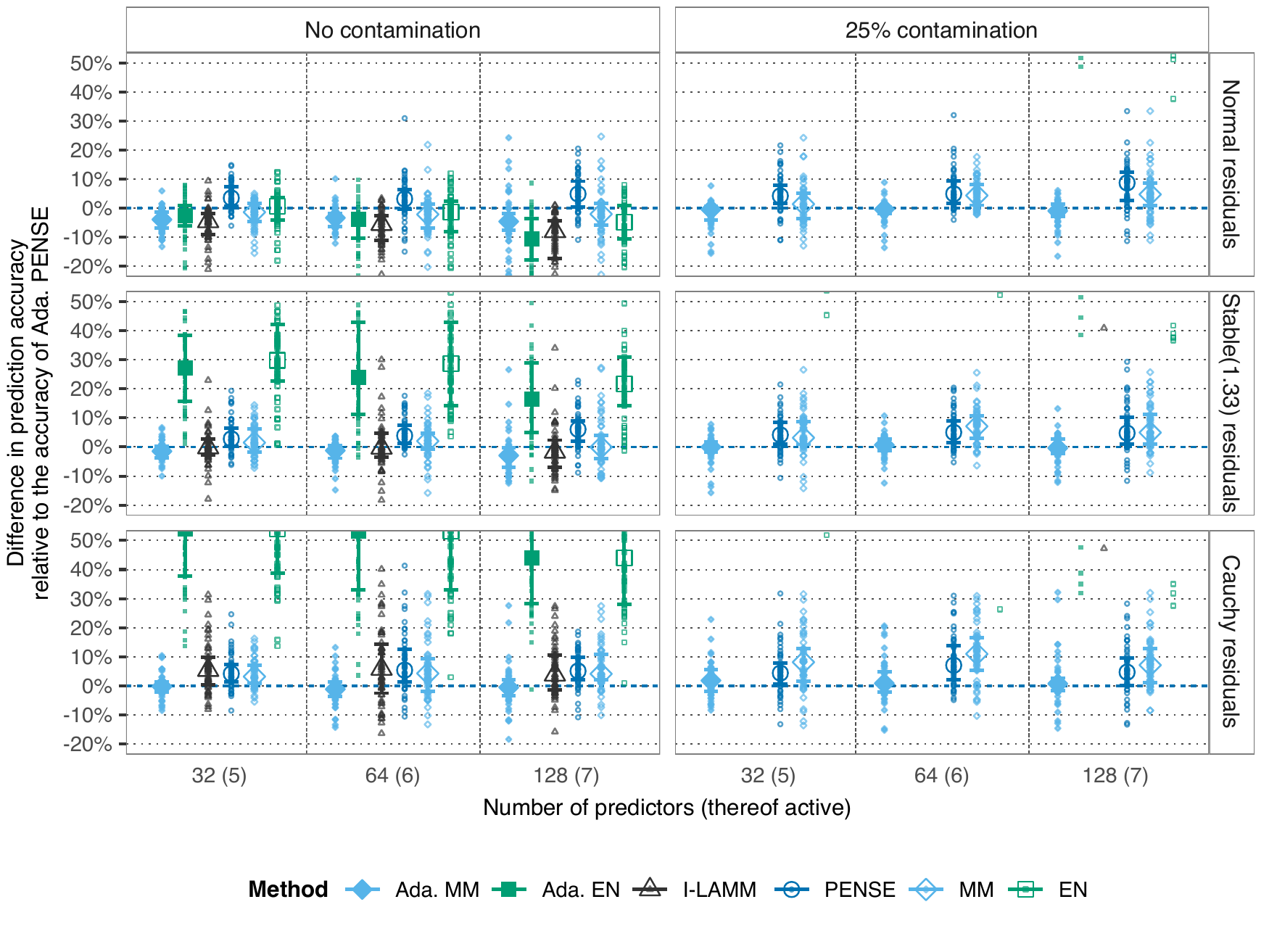}

}
\caption{%
Prediction accuracy of robust and non-robust estimators in the alternative scenario, relative to the prediction accuracy of adaptive PENSE.
The relative accuracy is defined in Equation~\eqref{eqn:simstudy-rel-pred_err}, with positive values indicating a larger scale of the prediction errors than achieved by adaptive PENSE and negative values indicating better prediction accuracy than adaptive PENSE.
The median out of 50 replications is depicted by the points and the lines show the interquartile range.}
\label{fig:simstudy-scenario_2-prediction_accuracy-relative}
\end{figure}

\begin{figure}[ht]
  {\centering \includegraphics[width=1\linewidth]{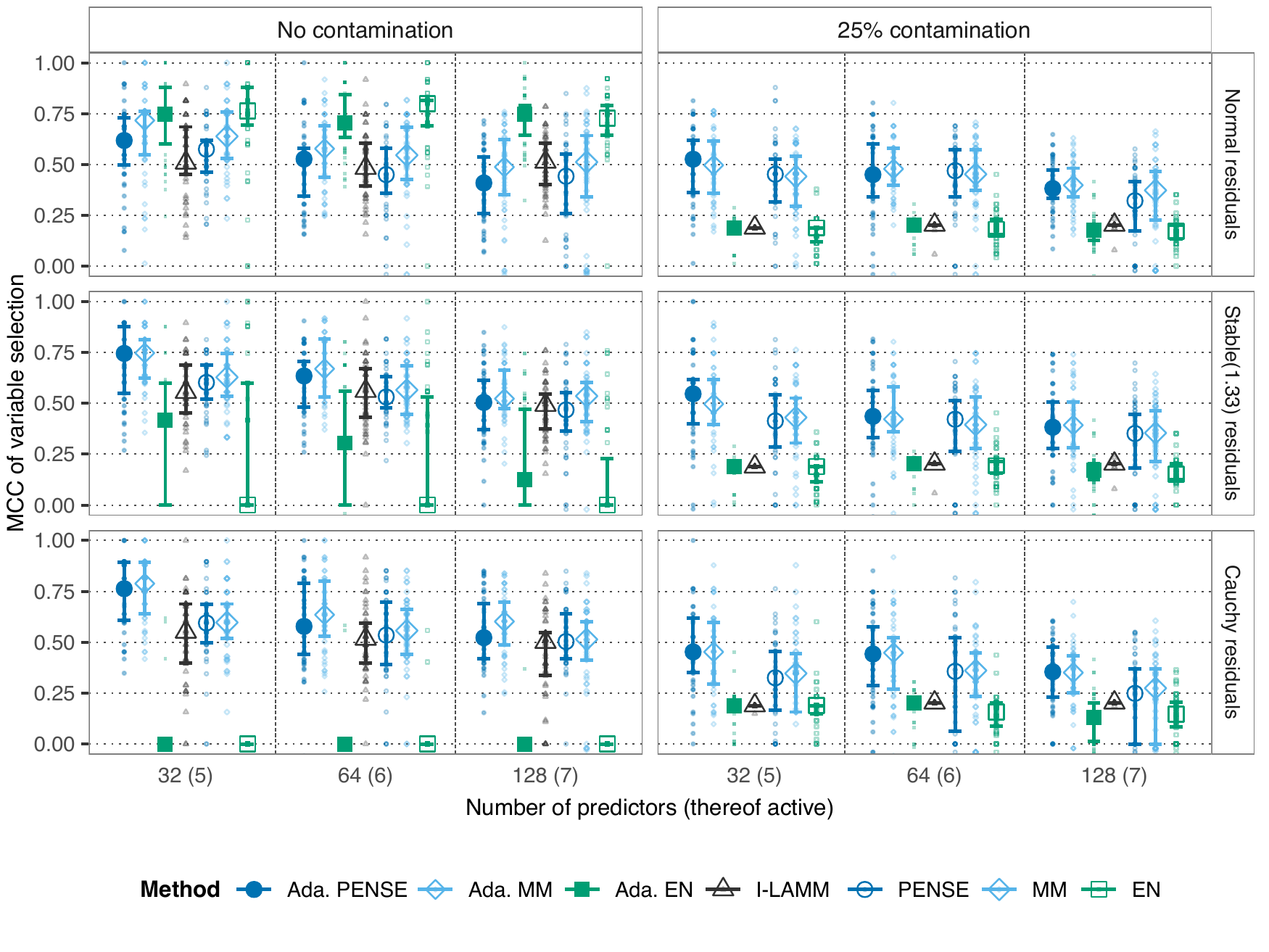}
  
  }
  \caption{%
  Overall variable selection performance of robust and non-robust estimators in the alternative scenario.
  Performance is measured by the Matthews correlation coefficient (MCC; higher is better), defined in~\eqref{eqn:simstudy-mcc}.
  The median out of 50 replications is depicted by the points and the lines show the interquartile range.}
  \label{fig:simstudy-scenario_2-mcc}
\end{figure}

\begin{figure}[ht]
  {\centering \includegraphics[width=1\linewidth]{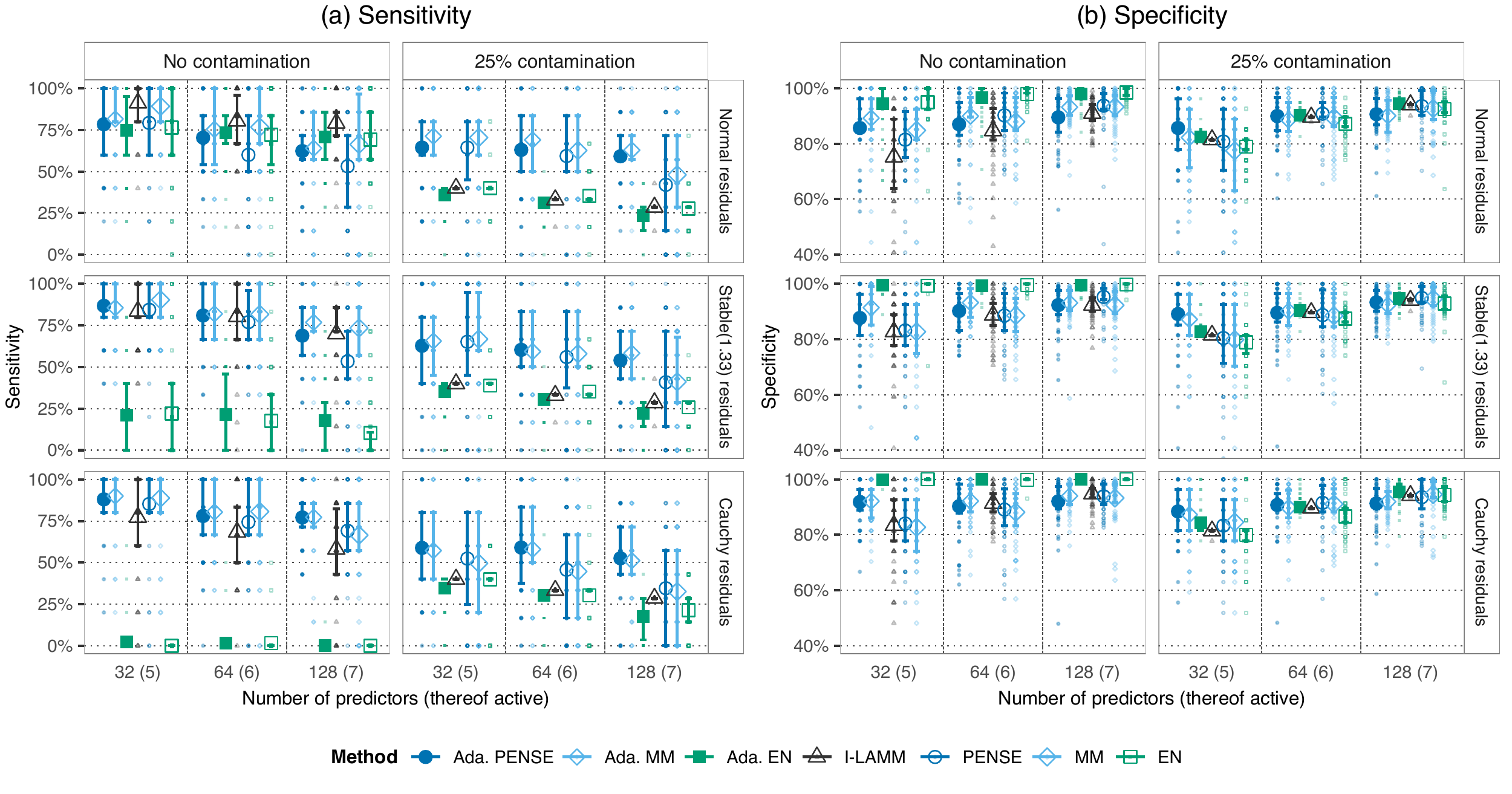}
  
  }
  \caption{%
  Detailed variable selection performance of robust and non-robust estimators in the alternative scenario.
  Sensitivity (left) is the number of selected predictors which are truly active relative to the total number of truly active predictors.
  Specificity (right) is the number of predictors which are truly inactive and not selected relative to the total number of inactive predictors.
  The median out of 50 replications is depicted by the points and the lines show the interquartile range.}
  \label{fig:simstudy-scenario_2-sens_spec}
\end{figure}